\documentclass[10pt,aps,onecolumn,tightenlines,superscriptaddress,notitlepage]{revtex4-2}

\newcommand{\vertiii}[1]{{\left\vert\kern-0.25ex\left\vert\kern-0.25ex\left\vert #1 
    \right\vert\kern-0.25ex\right\vert\kern-0.25ex\right\vert}}

\newcommand{\serie}[1]{\{#1_{n}\}_n}

\usepackage{newpxtext,newpxmath}

\let\coloneqq\relax

\usepackage{amsthm}
\usepackage{amssymb}
\usepackage{amsmath}
\usepackage{bbold}
\usepackage{bbm}
\usepackage{nonfloat}
\usepackage[pdftex]{hyperref}
\usepackage{braket}
\usepackage{dsfont}
\usepackage{mathdots}
\usepackage{mathtools}
\usepackage{enumerate}
\usepackage[shortlabels]{enumitem}
\usepackage{csquotes}
\usepackage{stmaryrd}
\usepackage[cal=boondox]{mathalfa}
\usepackage{graphicx}
\usepackage{stackengine}
\usepackage{scalerel}
\usepackage{xr}
\usepackage{array}
\usepackage{makecell}
\newcolumntype{x}[1]{>{\centering\arraybackslash}p{#1}}
\usepackage{tikz}
\usepackage{pgfplots}
\usetikzlibrary{shapes.geometric, shapes.misc, positioning, arrows, arrows.meta, decorations.pathreplacing, decorations.pathmorphing, patterns, angles, quotes, calc}
\usepackage{booktabs}
\usepackage{xfrac}
\usepackage{siunitx}
\usepackage{centernot}
\usepackage{comment}
\usepackage{chngcntr}

\newtheorem{thm}{Theorem}
\newtheorem*{thm*}{Theorem}

\newtheorem*{prop*}{Proposition}
\newtheorem{lemma}[thm]{Lemma}
\newtheorem*{lemma*}{Lemma}

\newtheorem*{cor*}{Corollary}

\newtheorem*{cj*}{Conjecture}
\newtheorem{Def}[thm]{Definition}
\newtheorem*{Def*}{Definition}

\newtheorem{remark}{Remark}

\makeatletter
\def\thmhead@plain#1#2#3{%
  \thmname{#1}\thmnumber{\@ifnotempty{#1}{ }\@upn{#2}}%
  \thmnote{ {\the\thm@notefont#3}}}
\let\thmhead\thmhead@plain
\makeatother

\theoremstyle{definition}

\newcommand{\bb}{\begin{equation}\begin{aligned}\hspace{0pt}}
\newcommand{\bbb}{\begin{equation*}\begin{aligned}}
\newcommand{\ee}{\end{aligned}\end{equation}}
\newcommand{\eee}{\end{aligned}\end{equation*}}
\newcommand*{\coloneqq}{\mathrel{\vcenter{\baselineskip0.5ex \lineskiplimit0pt \hbox{\scriptsize.}\hbox{\scriptsize.}}} =}

\newcommand{\ketbra}[1]{\ket{#1}\!\!\bra{#1}}

\newcommand{\N}{\mathds{N}}
\newcommand{\C}{\mathds{C}}

\DeclareMathOperator{\Tr}{Tr}
\DeclareMathOperator{\rk}{rk}

\DeclareMathAlphabet{\pazocal}{OMS}{zplm}{m}{n}

\DeclareMathOperator{\Id}{Id}

\newcommand{\NN}{\mathcal{N}}

\newcommand{\lsmatrix}{\left(\begin{smallmatrix}}
\newcommand{\rsmatrix}{\end{smallmatrix}\right)}

\stackMath

\stackMath

\makeatletter
\newcommand*\rel@kern[1]{\kern#1\dimexpr\macc@kerna}
\newcommand*\widebar[1]{%
  \begingroup
  \def\mathaccent##1##2{%
    \rel@kern{0.8}%
    \overline{\rel@kern{-0.8}\macc@nucleus\rel@kern{0.2}}%
    \rel@kern{-0.2}%
  }%
  \macc@depth\@ne
  \let\math@bgroup\@empty \let\math@egroup\macc@set@skewchar
  \mathsurround\z@ \frozen@everymath{\mathgroup\macc@group\relax}%
  \macc@set@skewchar\relax
  \let\mathaccentV\macc@nested@a
  \macc@nested@a\relax111{#1}%
  \endgroup
}

\counterwithin*{equation}{part}
\counterwithin*{thm}{part}
\counterwithin*{figure}{part}

\tikzset{meter/.append style={draw, inner sep=10, rectangle, font=\vphantom{A}, minimum width=30, line width=.8, path picture={\draw[black] ([shift={(.1,.3)}]path picture bounding box.south west) to[bend left=50] ([shift={(-.1,.3)}]path picture bounding box.south east);\draw[black,-latex] ([shift={(0,.1)}]path picture bounding box.south) -- ([shift={(.3,-.1)}]path picture bounding box.north);}}}
\tikzset{roundnode/.append style={circle, draw=black, fill=gray!20, thick, minimum size=10mm}}
\tikzset{squarenode/.style={rectangle, draw=black, fill=none, thick, minimum size=10mm}}

\definecolor{Blues5seq1}{RGB}{239,243,255}
\definecolor{Blues5seq2}{RGB}{189,215,231}
\definecolor{Blues5seq3}{RGB}{107,174,214}
\definecolor{Blues5seq4}{RGB}{49,130,189}
\definecolor{Blues5seq5}{RGB}{8,81,156}

\definecolor{Greens5seq1}{RGB}{237,248,233}
\definecolor{Greens5seq2}{RGB}{186,228,179}
\definecolor{Greens5seq3}{RGB}{116,196,118}
\definecolor{Greens5seq4}{RGB}{49,163,84}
\definecolor{Greens5seq5}{RGB}{0,109,44}

\definecolor{Reds5seq1}{RGB}{254,229,217}
\definecolor{Reds5seq2}{RGB}{252,174,145}
\definecolor{Reds5seq3}{RGB}{251,106,74}
\definecolor{Reds5seq4}{RGB}{222,45,38}
\definecolor{Reds5seq5}{RGB}{165,15,21}

\usepackage{pgfplots}

\pgfplotsset{width=10cm,compat=1.9}
\usetikzlibrary{decorations.markings}

\definecolor{marco}{rgb}{.4,.2,.6}

\allowdisplaybreaks

\newcommand*{\addFileDependency}[1]{
  \typeout{(#1)}
  \@addtofilelist{#1}
  \IfFileExists{#1}{}{\typeout{No file #1.}}
}
\makeatother

\usepackage{algorithm}
\usepackage{algpseudocode}
\algrenewcommand\algorithmicrequire{\textbf{Input:}}
\algrenewcommand\algorithmicensure{\textbf{Output:}}

\begin{document}

\author{Francesco Anna Mele}
\email{francesco.mele@sns.it}
\affiliation{NEST, Scuola Normale Superiore and Istituto Nanoscienze, Piazza dei Cavalieri 7, IT-56126 Pisa, Italy}

\author{Giovanni Barbarino}
\email{giovanni.barbarino@gmail.com}
\affiliation{Mathematics and Operational Research Unit, Faculté Polytechnique de Mons, UMONS, BE-7000 Mons, Belgium}

\author{Vittorio Giovannetti}
\email{vittorio.giovannetti@sns.it}
\affiliation{NEST, Scuola Normale Superiore and Istituto Nanoscienze, Consiglio Nazionale delle Ricerche, Piazza dei Cavalieri 7, IT-56126 Pisa, Italy}

\author{Marco Fanizza}
\email{marco.fanizza@uab.cat}
\affiliation{{F\'{\i}sica Te\`{o}rica: Informaci\'{o} i Fen\`{o}mens Qu\`{a}ntics, Departament de F\'{i}sica, Universitat Aut\`{o}noma de Barcelona, ES-08193 Bellaterra (Barcelona), Spain}}

\title{Non-asymptotic quantum communication on lossy transmission lines with memory}

\begin{abstract}
Non-asymptotic quantum Shannon theory analyses how to transmit quantum information across a quantum channel as efficiently as possible within a specified error tolerance, given access to a finite, fixed, number of channel uses. In a recent work, we derived computable lower bounds on the non-asymptotic capacities of \emph{memoryless} bosonic Gaussian channels. In this work, we extend these results to the \emph{non-Markovian} bosonic Gaussian channel introduced in F. A. Mele, G. D. Palma, M. Fanizza, V. Giovannetti, and L. Lami 
IEEE Transactions on Information Theory {\bf 70}(12), 8844-8869 (2024), which describes non-Markovian effects in optical fibres and is a non-Markovian generalisation of the pure loss channel. This allows us to determine how many uses of a non-Markovian optical fibre are sufficient in order to transmit $k$ qubits, distil $k$ ebits, or generate $k$ secret-key bits up to a given error tolerance $\varepsilon$. To perform our analysis, we prove novel properties of singular values of Toeplitz matrices, providing an error bound on the convergence rate of the celebrated Avram–Parter’s theorem, which we regard as a new tool of independent interest for the field of quantum information theory and matrix analysis.
\end{abstract}

\maketitle  

\section{Introduction}
Continuous-variable quantum communication is expected to play a pivotal role in the evolution of future quantum technologies~\cite{usenko2025continuousvariablequantumcommunication}. It enables a range of critical capabilities, including secure information transfer, qubit transmission, and entanglement distribution via free-space channels and optical fibres~\cite{BUCCO, Caves, weedbrook12, Pirandola20,usenko2025continuousvariablequantumcommunication}. Among these applications, continuous-variable quantum key distribution stands out as a mature technology, propelled by significant experimental progress in recent years~\cite{Record1, Record2, Record3, Record4, Record5}.
Traditionally, the analysis of bosonic quantum communication is conducted under the so-called \emph{memoryless assumption} (or \emph{i.i.d.~assumption})~\cite{memory-review,holwer,Caruso2006, Wolf2007, Mark2012,Mark-energy-constrained, Rosati2018, Sharma2018, Noh2019, Noh2020,fanizza2021estimating,lower-bound,Giova_classical_cap,LossyECEAC1, LossyECEAC2, PLOB, Davis2018, Goodenough16, TGW, MMMM, squashed_channel, Pirandola2009, Ottaviani_new_lower, Pirandola18,LL-bosonic-dephasing,Mele_2024}, which posits that the noise affecting the communication channel is independent of the input signals previously transmitted. While this assumption provides a valid approximation when signals are transmitted with sufficiently long time intervals, it breaks down when the temporal rate of signal transmission increases~\cite{memory-review}. In such high-rate scenarios, the noise affecting the channel depends on previous signals, giving rise to \emph{memory} (non-Markovian) effects~\cite{memory-review,Memory1,Memory2,Memory3,Die-Hard-2-PRA,Die-Hard-2-PRL,mele2023optical}. 
 In recent years, the study of memory effects in quantum communication has garnered significant attention due to its potential to optimise the performance of optical channels~\cite{memory-review,Memory1,Memory2,Memory3,Die-Hard-2-PRA,Die-Hard-2-PRL,mele2023optical}. These memory effects can profoundly influence critical channel properties, including its communication capabilities and noise dynamics. Consequently, a thorough understanding and accurate modeling of these phenomena are critical for enhancing the transmission rates of communication lines and for designing robust quantum communication protocols tailored to practical, real-world environments.

Prior to this work, the communication performances of bosonic Gaussian channels with memory effects have been evaluated only through their \emph{asymptotic capacities}~\cite{memory-review,Memory1,Memory2,Memory3,Die-Hard-2-PRA,Die-Hard-2-PRL,mele2023optical}, which quantify the highest ratio of extracted resources (e.g.~transmitted bits, qubits, Bell pairs, or secret-key bits) to the number of channel uses, assuming an asymptotic regime of infinite channel uses and negligible errors~\cite{Sumeet_book}. However, no progress has been made in the realm of \emph{non-asymptotic} (or \emph{one-shot}) quantum Shannon theory~\cite{Tomamichel2015, Sumeet_book, Tomamichel2012, Tomamichel2016, Tomamichel2008, Berta2011, cheng2024invitationsamplecomplexityquantum,  MMMM, Kaur_2017, khatri2021secondorder, WildeRenes2016,mele2025achievableratesnonasymptoticbosonic} under bosonic noise with memory effects.
In this non-asymptotic framework, the concept of asymptotic capacity is replaced by the one of \emph{non-asymptotic capacity}~\cite{Sumeet_book}, which measures the maximum number of resources --- such as bits, qubits, Bell pairs, or secret-key bits --- that can be reliably extracted with a specified error tolerance after a finite number of channel uses~\cite{MMMM, Kaur_2017, khatri2021secondorder, WildeRenes2016, LL-bosonic-dephasing}. This approach is crucial, as it provides practical insights into the waiting times required to complete specific quantum communication tasks. Additionally, it parallels the study of "sample complexity" or "query complexity," which are established concepts in quantum learning theory~\cite{anshu2023survey, mele2024learningquantumstatescontinuous,fanizza2024efficienthamiltonianstructuretrace,bittel2024optimalestimatestracedistance}.

In this work, we fill the gap of estimating  non-asymptotic capacities of relevant Gaussian channels with memory effects. Specifically, we derive easily computable lower bounds on the non-asymptotic capacities of the bosonic Gaussian channel model introduced in~\cite{mele2023optical}, which captures memory effects in long-distance optical communication links. To do this, we build on recent findings regarding non-asymptotic capacities of memoryless Gaussian channels~\cite{mele2025achievableratesnonasymptoticbosonic}. A key component of our analysis is the development of novel properties related to the singular values of Toeplitz matrices~\cite{Szego1920}, which allow us to rigorously bound the error in the convergence rate of the so-called \emph{Avram-Parter's theorem}~\cite{Szego1920, GRENADER, Avram1988, Parter1986}, and, in particular, of the \emph{Szeg\"o theorem}~\cite{Szego1920} --- both celebrated theorems in the matrix analysis literature. This not only serves as a crucial tool for our analysis of non-asymptotic capacities but also represents a remarkable advancement in matrix analysis, with potential applications extending to broader areas of quantum information theory.

The paper is structured as follows. In Section~\ref{sec_notation}, first we introduce the necessary notation and preliminaries to support the presentation of our results, and second we review the model of optical fibre with memory effects introduced in \cite{mele2023optical}. In Section~\ref{Sec_memory_main}, we present our main results: (i) we provide easily computable lower bounds on the non-asymptotic capacities of the model; (ii) we derive an explicit bound on the convergence of the Avram--Parter's theorem~\cite{Szego1920, GRENADER,Avram1988,Parter1986}, which turns out to be crucial to bound the non-asymptotic capacities. Finally, in the Appendix, we provide the complete proofs of all results presented in the paper.

\section{Preliminaries}\label{sec_notation}
Here, we introduce the relevant notation, review some basic facts about bosonic Gaussian channels, and introduce the model of optical fibre with memory effects of Ref.~\cite{mele2023optical}.

\subsection{Non-asymptotic capacities of memoryless channels}
In the context of quantum Shannon theory~\cite{Sumeet_book,MARK}, the noise affecting a memoryless communication line can be modelled via a single quantum channel $\NN$, which acts on each of the transmissions signals. Specifically, a quantum channel $\NN$ is a linear, completely positive, trace-preserving transformation from the set of input quantum states to the set of output quantum states~\cite{Sumeet_book,MARK}. Additionally, $n$ uses of such a memoryless communication line are modelled via the tensor-product quantum channel $\NN^{\otimes n}$. More precisely, any (possibly correlated) input state $\rho^{(n)}$ of the first $n$ signals are mapped to the (potentially degraded) output state via the following relation:
\begin{eqnarray} 
\rho^{(n)} \mapsto \NN^{\otimes n} \!\left(\rho^{(n)}\right)\;. \label{mapping} 
\end{eqnarray} 
The quality of the quantum channel $\NN$ can be characterised using figures of merit tailored to the specific communication task under consideration. In this work, we focus on three communication tasks: 
 \begin{itemize}
 	\item \emph{Qubit distribution}: The goal is to reliably transmit an arbitrary quantum state, possibly entangled with an auxiliary system, across the quantum channel;
 	\item \emph{Entanglement distribution assisted by two-way classical communication}: The goal is to generate two-qubit Bell pairs (referred to as \emph{ebits}) between two parties connected by the quantum channel, exploiting the additional resource of a two-way classical communication line;
 	\item \emph{Secret-key distribution assisted by two-way classical communication}: The goal is to generate secret keys shared between two parties linked by the quantum channel, exploiting the additional resource of a two-way classical communication line;
 \end{itemize}
For each of these communication tasks, one can define a proper notion of \emph{non-asymptotic capacity}~\cite{Sumeet_book,mele2025achievableratesnonasymptoticbosonic}:
\begin{itemize}
    \item the \emph{$n$-shot quantum capacity $Q^{(\varepsilon,n)}(\NN)$} is defined as the maximum number of qubits that can be transmitted up to error $\varepsilon$ across $n$ uses of $\NN$;
    \item the \emph{$n$-shot two-way quantum capacity $Q_2^{(\varepsilon,n)}(\NN)$} is defined as the maximum number of ebits  that one can distil up to error $\varepsilon$ using $n$ uses of $\NN$ and arbitrary LOCCs (Local Operation and Classical Communication operations) ;
 \item the \emph{$n$-shot secret-key capacity $K^{(\varepsilon,n)}(\NN)$} is defined as the maximum number of secret-key bits that one can distil up to error $\varepsilon$ using $n$ uses of $\NN$ and arbitrary LOCCs.
\end{itemize}
For the exact definitions of the errors $\varepsilon$ associated with each of the these quantum communication tasks, as well as rigorous definitions of the capacities $Q^{(\varepsilon,n)}(\NN)$, $Q_2^{(\varepsilon,n)}(\NN)$, $K^{(\varepsilon,n)}(\NN)$,
 we refer to Ref.~\cite{mele2025achievableratesnonasymptoticbosonic}.

The \emph{asymptotic capacities} are defined by taking the limit of large channel uses ($n\rightarrow\infty$) and vanishing error ($\varepsilon\rightarrow0$) of the regularised version of their non-asymptotic counterparts~\cite{MARK,Sumeet_book}. That is, the \emph{quantum capacity} $Q(\NN)$, \emph{two-way quantum capacity} $Q_2(\NN)$, and \emph{secret-key capacity} $K(\NN)$ are defined as
  \bb
        Q(\NN)&\coloneqq \lim\limits_{\varepsilon \rightarrow 0^+} \liminf_{n \to \infty}\frac{Q^{(\varepsilon,n)}(\NN)}{n}\,,\\
        Q_2(\NN)&\coloneqq \lim\limits_{\varepsilon \rightarrow 0^+} \liminf_{n \to \infty}\frac{Q_2^{(\varepsilon,n)}(\NN)}{n}\,,\\
        K(\NN)&\coloneqq \lim\limits_{\varepsilon \rightarrow 0^+} \liminf_{n \to \infty}\frac{K^{(\varepsilon,n)}(\NN)}{n}\,.
  \ee
  In other words, the quantum capacity of $\NN$ is the maximum achievable rate of qubits which can be transmitted with vanishing error in the asymptotic limit of infinitely many uses of the channel. Analogously, the two-way quantum capacity (resp.~secret-key capacity) of $\NN$ is the maximum achievable rate of ebits (resp.~secret-key bits) which can be distilled via LOCCs with vanishing error in the asymptotic limit of infinitely many uses of the channel. Note that a simple resource counting argument establishes the following ordering among the non-asymptotic capacities:
  \bb 
  Q^{(\varepsilon,n)}(\NN) &\leq Q_2^{(\varepsilon,n)}(\NN)\leq K^{(\varepsilon,n)}(\NN)\;, 
  \ee
  and, in particular, the same holds for their asymptotic counterparts, i.e.
  \bb
  Q(\NN) &\leq Q_2(\NN)\leq K(\NN)\;.
  \ee
 \subsection{Non-asymptotic capacities of non-Markovian channels}
When memory effects are present, the noise affecting a given input signal depends on the specific state of the previously sent signals~\cite{memory-review,Memory1,Memory2,Memory3,Die-Hard-2-PRA,Die-Hard-2-PRL,mele2023optical}. In this context, the input-output relations are no longer governed by a single quantum channel $\NN$. Instead, they are described by a family of quantum channels  \begin{eqnarray} \label{family} {\cal F}\coloneqq\{ \NN^{(1)},\NN^{(2)}, \cdots, \NN^{(n)}, \cdots\}\;,\end{eqnarray} 
 where $\NN^{(n)}$ replaces $\NN^{\otimes n}$ in Eq.~(\ref{mapping}) to model the input-output
 mapping of the first $n$ signals: 
\begin{eqnarray} 
\rho^{(n)} \mapsto \NN^{(n)} (\rho^{(n)})\; ,\label{mapping1} 
\end{eqnarray} 
where $\rho^{(n)} $ denotes the initial state of the first $n$ signals. Although the expressions of the various maps in ${\cal F}$ depend upon the specific details of the communication line, they must always satisfy the following consistency requirement, which ensures that future transmissions do not affect past signaling events: 
\begin{eqnarray}  \NN^{(n)} = {\mathcal T}_{n+1}\circ \NN^{(n+1)}\qquad \forall n\in\mathbb{N}\,, \end{eqnarray} 
where ${\cal T_{n+1}}$ represents the partial trace with respect to the $(n+1)$-th output signal. In this non-Markovian context, the $n$-shot capacities of the communication line can be 
formally identified as the one-shot capacities of the quantum channel $\NN^{(n)}$, i.e. 
\bb \label{defmem-ncap} 
    Q^{(\varepsilon,n)}\left({\cal F}\right)&\coloneqq Q^{(\varepsilon,1)}\!\left(\NN^{(n)}\right)\,, \\
      Q_2^{(\varepsilon,n)}\!\left({\cal F}\right)&\coloneqq Q_2^{(\varepsilon,1)}\!\left(\NN^{(n)}\right)\,,\\
     K^{(\varepsilon,n)}\!\left({\cal F}\right)&\coloneqq K^{(\varepsilon,1)}\left(\NN^{(n)}\right)\,.
\ee
Additionally, as before, the asymptotic capacities are defined in terms of the regularised version of the corresponding non-asymptotic capacities in the limit of
large $n$ and vanishing $\varepsilon$:
\bb
    Q\left({\cal F}\right)&\coloneqq
    \lim\limits_{\varepsilon \rightarrow 0^+} \liminf_{n \to \infty} \frac{Q^{(\varepsilon,n)}\left({\cal F}\right)}{n}\,,\\
    Q_2\left({\cal F}\right)&\coloneqq
    \lim\limits_{\varepsilon \rightarrow 0^+} \liminf_{n \to \infty} \frac{Q_2^{(\varepsilon,n)}\left({\cal F}\right)}{n}\,,\\
    K\left({\cal F}\right)&\coloneqq
    \lim\limits_{\varepsilon \rightarrow 0^+} \liminf_{n \to \infty} \frac{K^{(\varepsilon,n)}\left({\cal F}\right)}{n}\,.
\ee

\subsection{Pure loss channel} 
Let us briefly give some preliminaries about continuous-variable systems; for further details, we refer to Ref.~\cite{BUCCO}. A continuous-variable system with $n$ modes is characterised by the Hilbert space $L^2(\mathbb{R}^n)$, which consists of all square-integrable, complex-valued functions over $\mathbb{R}^n$~\cite{BUCCO}. In quantum optics, such systems are commonly used to describe the electromagnetic field~\cite{BUCCO, Caves}. By definition, a state is said to be \emph{Gaussian} if it is a Gibbs state of a quadratic hamiltonian in the position and momentum operators of the system~\cite{BUCCO}. Moreover, a \emph{bosonic 
Gaussian channel}~\cite{BUCCO} is a quantum channel that maps the set of Gaussian states into itself.

The most important Gaussian channel is arguably the \emph{pure loss channel}, a single-mode bosonic Gaussian channel that models the phenomenon of photon loss in memoryless optical links~\cite{BUCCO}. The pure loss channel, denoted as $\mathcal{E}_{\lambda}$, is characterised by a single parameter $\lambda\in[0,1]$, which represents the ratio between the output energy and the input energy of the transmitted signals. 
Mathematically, $\mathcal{E}_{\lambda}$ acts on an input single-mode state $\rho$ by mixing it in a beam splitter (i.e.~a semi-reflective mirror) with an environmental mode ${E}$ initialised in the vacuum state $|0\rangle_{E}$:
\bb
        \mathcal{E}_{\lambda}(\rho)&\coloneqq\Tr_E\left[U_\lambda^{SE} \big(\rho \otimes\ketbra{0}_{E}\big) (U_\lambda^{SE})^\dagger\right] \,,\label{deflossy} 
\ee
where  $U_{\lambda}\coloneqq\exp\!\left[\arccos\sqrt{\lambda}\left(a^\dagger b-a\, b^\dagger\right)\right]$ denotes the beam splitter unitary of transmissivity $\lambda$, with $a$ and $b$ being the annihilation operators of the  signal and of the environment, respectively, and where $\Tr_{E}$ denotes the partial trace with respect to the environment --- see Fig.~\ref{fig_pure_loss} for a pictorial representation.
\begin{figure}[h!]
	\centering
	\includegraphics[width=0.3\linewidth]{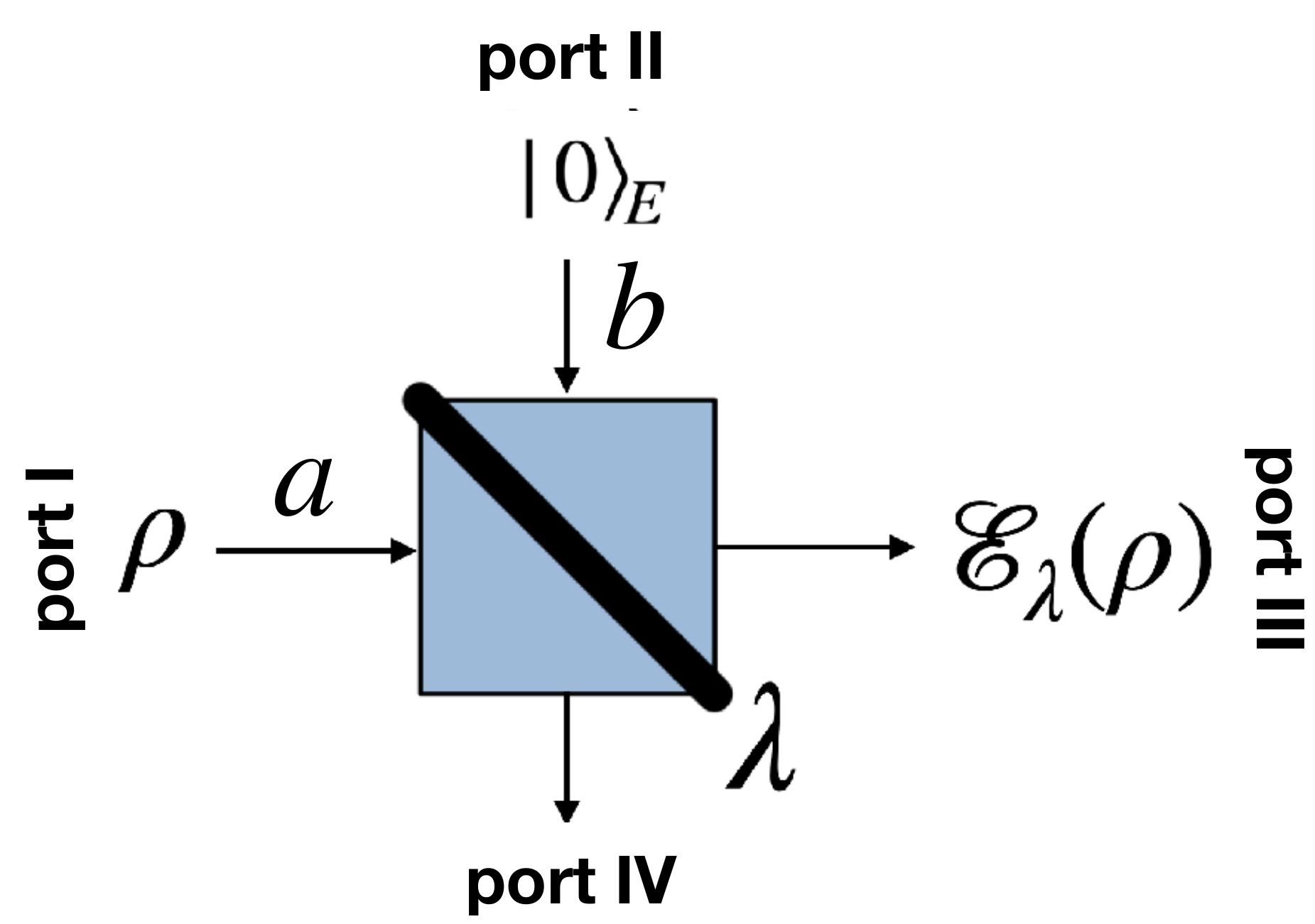}
	\caption{Schematic of the beam splitter representation~(\ref{deflossy}) of the pure loss channel 
	$\mathcal{E}_\lambda$. A beam splitter of a transmissivity $\lambda$ is a two-inputs/two-outputs device, that coherently mixes an incoming signal state $\rho$ entering from port I, with the vacuum state $|0\rangle_E$ of the environmental mode $E$, entering from port II. The transmitted signal emerges from port III in the final state  $\mathcal{E}_{\lambda}(\rho)$ obtained by taking the partial trace with respect to $E$ of the state $U_\lambda^{SE} \big(\rho \otimes\ketbra{0}_{E}\big) (U_\lambda^{SE})^\dagger$. The modified state of the environment after the interaction with the signal emerges instead from port IV.  }
    \label{fig_pure_loss}
\end{figure}

Note that the pure loss channel $\mathcal{E}_{\lambda}$ is noiseless for $\lambda=1$ (it equals the identity channel), while it is completely noisy for $\lambda=0$ (it maps any input state into the vacuum). It is also worth recalling that 
 the pure loss channel satisfies the following composition rule (see e.g.~\cite{Die-Hard-2-PRA}):
\bb 
\mathcal{E}_{\lambda_1}\circ \mathcal{E}_{\lambda_2}=\mathcal{E}_{\lambda_1\lambda_2}\qquad\forall\lambda_1,\lambda_2\in[0,1]\;. \label{concatenation}
\ee
Both the asymptotic and non-asymptotic capacities of the pure loss channel have been extensively analysed~\cite{holwer, Wolf2007, Wolf2006,PLOB,MMMM,mele2025achievableratesnonasymptoticbosonic}. Specifically, the (asymptotic) quantum capacity, two-way quantum capacity, and secret-key capacity of the pure loss channel are given by the following simple expressions:~\cite{holwer, Wolf2007, Wolf2006,PLOB}
 \bb\label{capacities_pure_loss_main}
    Q(\mathcal{E}_{\lambda})&=   
        \begin{cases}
        \log_2\!\left(\frac{\lambda}{1-\lambda}\right) &\text{if $\lambda\in(\frac{1}{2},1]$ ,} \\
        0 &\text{if $\lambda\in[0,\frac{1}{2}]$ ,}
    \end{cases} \\
    Q_2(\mathcal{E}_{\lambda})&=K(\mathcal{E}_{\lambda})= \log_2\!\left(\frac{1}{1-\lambda}\right)\,,
\ee
Moreover, given $\varepsilon\in(0,1)$ and $n\in\mathbb{N}$, the $n$-shot two-way quantum capacity and the $n$-shot secret-key capacity can be lower bounded as~\cite{mele2025achievableratesnonasymptoticbosonic}
\bb\label{ineq_best_bounds_main02}
K^{(\varepsilon,n)}(\mathcal{E}_{\lambda})&\ge Q_2^{(\varepsilon,n)}(\mathcal{E}_{\lambda}) \ge  n Q_2(\mathcal{E}_{\lambda})
- \log_2\!\left(\frac{2^{6}\,3\,(4-\sqrt{\varepsilon})^2}{(2-\sqrt{\varepsilon})\varepsilon^3 }\right)\,,
\ee
and upper bounded as~\cite{MMMM}
\bb\label{upper_bound_n_shotmain}
Q_2^{(\varepsilon,n)}(\mathcal{E}_\lambda)\le K^{(\varepsilon,n)}(\mathcal{E}_\lambda)\le n Q_2(\mathcal{E}_{\lambda})+\log_26+2\log_2\!\left(\frac{1+\varepsilon}{1-\varepsilon}\right)\,.
\ee
Additionally, the $n$-shot quantum capacity can be lower bounded as~\cite{mele2025achievableratesnonasymptoticbosonic}
\bb\label{ineq_best_bounds_main01}
    Q^{(\varepsilon,n)}(\mathcal{E}_{\lambda}) 
    &\geq  
    n Q(\mathcal{E}_{\lambda}) - \log_2\!\left( \frac{2^{23}(32-\varepsilon)^2}{ (16-\varepsilon)\varepsilon^6}\right)\;,
\ee
and, for $\varepsilon\in(0,\frac{1}{2})$, it can be upper bounded as
\bb
Q^{(\varepsilon,n)}(\mathcal{E}_\lambda)\le n \frac{Q(\mathcal{E}_{\lambda})}{(1-2\varepsilon)}-\varepsilon \log_2\varepsilon-(1-\varepsilon) \log_2(1-\varepsilon)\,,
\ee
where the latter follows from the fact that the pure loss channel is degradable, together with a general upper bound of the one-shot quantum capacity in terms of the coherent information~\cite[Corollary 14.4]{khatri2021secondorder}.

\subsection{Non-Markovian generalisation of the the pure loss channel} \label{Sec_memory}
In the absence of memory effects, an optical fibre is commonly modelled by the pure loss channel ${\mathcal E}_\lambda$, where the parameter $\lambda$ coincides with the effective transmissivity of the fibre. Specifically, the transmission of $n$ signals through the fibre is modelled as in Eq.~\eqref{mapping}, where the channel $\NN$ is identified with ${\mathcal E}_\lambda$. The first attempts to describe memory effects in optical fibres arose in Refs.~\cite{memory-review,Memory1,Memory2,Memory3}. Specifically, the models presented in Refs.~\cite{memory-review,Memory1,Memory2,Memory3} were based on a suitable non-Markovian generalisation of the pure loss channel~\cite{memory-review,Memory1,Memory2,Memory3}, which tried to describe memory effects by allowing for interferences between the state emerging from the environmental output port (i.e., port IV of Fig.~\ref{fig_pure_loss}) of the beam splitter associated with a given input signal and the states of the subsequent input signals~\cite{memory-review,Memory1,Memory2,Memory3}. However, as noted in Ref.~\cite{mele2023optical}, these models are not fully satisfactory as they fail to describe memory effects for long optical fibres; instead, they constitute a good model just for very short optical fibres.

To overcome these limitations, Ref.~\cite{mele2023optical} introduced a more refined model of optical fibres with memory effects. The peculiar feature of this model is that the interferences between signals occurs continuously along the entire fibre, making it a good model even for arbitrarily long fibres. In order to explain this more refined model, it is necessary to start from the following simple, provoking question: Why the pure loss channel is a consistent model of memoryless optical fibres?

At a first glance, one may be tempted to think that modelling a memoryless optical fibre with a pure loss channel is a too crude approximation, as it would mean to model a possibly very long line of glass (the fibre) via a single short mirror (the beam splitter). Following the same intuition, the pure loss channel appears to be a reasonable model when the fibre is very short (ideally infinitesimal), as the latter essentially becomes a small piece of glass, i.e.~a single beam splitter. To see why the pure loss channel $\mathcal{E}_\lambda$ is a good model even for long (memoryless) optical fibres, let us consider an optical fibre of transmissivity $\lambda$. Let us divide such fibre into $M$ segments, each of length $L/M$, where $L$ denotes the total lenght of the fibre. Since the transmissivity $\lambda$ is exponentially decreasing in the length $L$, each segment has a transmissivity of $\lambda^{1/M}$. As $M$ approaches infinity, each segment becomes infinitesimally short and thus can be modelled as a beam splitter of transmissivity $\lambda^{1/M}$, coupling the input signals traveling through the fibre with a (single-mode) localised environment initialised in the vacuum state (see Fig.~\ref{fig_memoryless_main} for a pictorial representation). Since the memoryless assumption guarantees that all the $M$ environments resets to the vacuum state after every channel use, it follows that the entire optical fibre can be modelled as the composition of $M\rightarrow\infty$ pure loss channels, each with transmissivity $\lambda^{1/M}$. By exploiting the composition rule in~\eqref{concatenation}, such a composition of $M$ pure loss channels turns out to be exactly equal to the pure loss channel~$\mathcal{E}_\lambda$. Hence, we conclude that the pure loss channel is a good model for \emph{memoryless} optical fibres.
\begin{figure}[h!]
	\centering
	\includegraphics[width=1\linewidth]{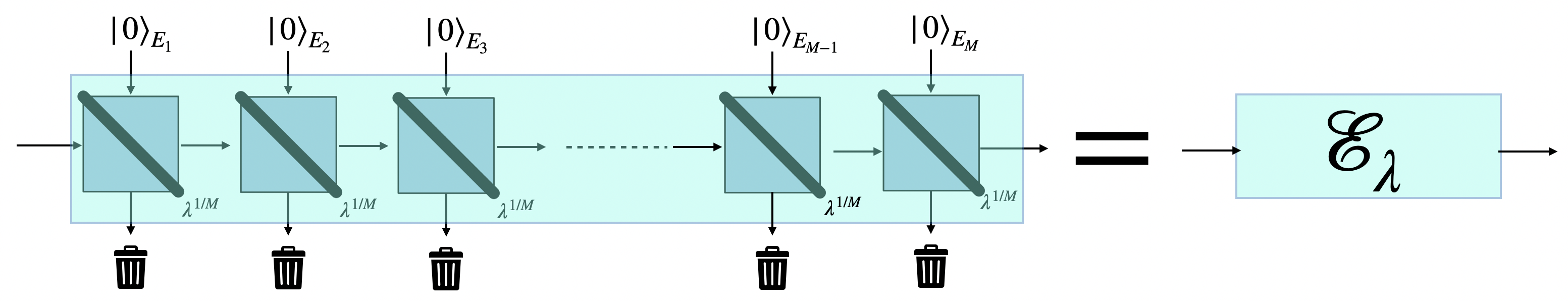}
	\caption{Pure loss channel $\mathcal{E}_\lambda$ as a model for a memoryless optical fibre with transmissivity $\lambda$. The fibre can be seen as the composition of $M$ optical fibres of transmissivity $\lambda^{1/M}$. In the limit $M\rightarrow\infty$, each of the $M$ infinitesimal optical fibres can be modelled as a pure loss channel of transmissivity $\lambda^{1/M}$. Due to the composition rule~(\ref{concatenation})  the entire optical fibre corresponds to a single pure loss channel of transmissivity $\lambda$.}
    \label{fig_memoryless_main}
\end{figure}

We are now ready to explain the model of optical fibre with memory effects introduced in Ref.~\cite{mele2023optical}.  The above-mentioned assumption that each environment resets to the vacuum state after every channel use is precisely what defines the memoryless (i.e.~Markovian) regime: the communication channel returns to its initial configuration after each use, meaning that all environments are restored to the vacuum. This resetting can be interpreted as a thermalisation process that counteracts the perturbation induced by the interaction between the input optical signals and the environment via the $M$ beam splitters. Under the memoryless assumption, this thermalisation process is perfect --- each environment is fully restored to the vacuum state after every channel use. However, when input signals arrive in rapid succession, thermalisation becomes imperfect, causing each new signal to interact with an $M$-mode environment that retains traces of previous signals instead of returning to the $M$-mode vacuum state. Since a pure loss channel provides an effective model for describing the thermalisation of a single-mode system undergoing energy dissipation~\cite{petruccione_book}, the authors of Ref.~\cite{mele2023optical}  model such a thermalisation process as a pure loss channel of transmissivity $\mu$ acting on each of the $M$ local environment after every channel use, as illustrated in Fig.~\ref{fig_memory_main}. When $\mu=0$, 
such a model reduces to the memoryless configuration described by the pure loss channel. As $\mu$ increases, the influence of previously transmitted signals on the optical link becomes more pronounced.

By summarising, the model of optical fibre with memory effects introduced in Ref.~\cite{mele2023optical} is characterised by two parameters: the transmissivity of the fibre $\lambda\in[0,1]$ and the memory parameter $\mu\in[0,1]$, which quantifies the extent of memory effects in the optical link.  In such a model, $n$ uses of the optical fibre are described by the $n$-mode Gaussian channel $\Phi^{(n)}_{\lambda,\mu}$ whose interferometric representation is depicted in Fig.~\ref{fig_memory_main} (for a precise mathematical definition of $\Phi^{(n)}_{\lambda,\mu}$, see Ref.~\cite{mele2023optical}). Using the notation introduced in Eq.~\eqref{family}, the communication line associated with this model turns out to be represended by the family of quantum channels 
\bb
	 {\cal F}_{\lambda,\mu}\coloneqq\{ \Phi^{(1)}_{\lambda,\mu},\Phi^{(2)}_{\lambda,\mu}, \cdots, \Phi^{(n)}_{\lambda,\mu}, \cdots\}\;,
\ee
where the channel $\Phi^{(n)}_{\lambda,\mu}$ evolves the first $n$ input signals. Finally, note that when the memory parameter vanishes ($\mu=0$), such a channel reduces to the memoryless configuration $\Phi^{(n)}_{\lambda,0}=\mathcal{E}_\lambda^{\otimes n}$ described by the pure loss channel  $\mathcal{E}_\lambda$ --- in this sense, such a model constitutes a non-Markovian generalisation of the pure loss channel.

\begin{figure}[h!]
	\centering
	\includegraphics[width=1\linewidth]{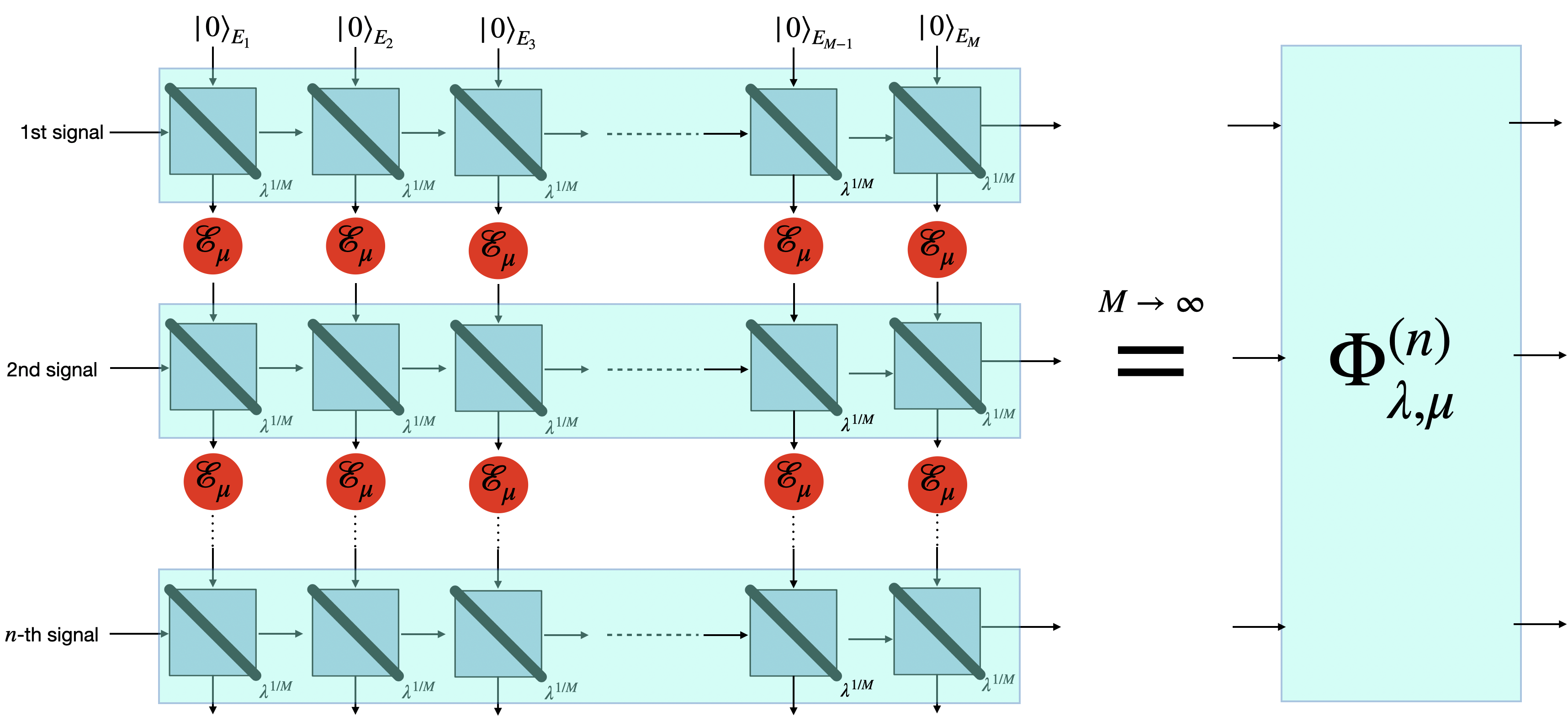}
	\caption{Interferometric representation of the
		the $n$-mode bosonic Gaussian channel $\Phi^{(n)}_{\lambda,\mu}$ used in 
		Ref.~\cite{mele2023optical} to describe memory effects on an optical fibre.	The fibre can be seen as the composition of $M$ optical fibres of transmissivity $\lambda^{1/M}$. In the limit $M\rightarrow\infty$, each of the $M$ infinitesimal optical fibres can be modelled as a beam splitter of transmissivity $\lambda^{1/M}$, which couples the input optical signal with a dedicated single-mode environment. After each use of the fibre, the $M$ environments undergo a thermalisation process, modelled by the pure loss channel $\mathcal{E}_\mu$, where the parameter $\mu$  is referred to as the memory parameter.  }
	\label{fig_memory_main}
\end{figure}

The asymptotic capacities of this model have been exactly determined~\cite{mele2023optical}. Specifically, the quantum capacity $Q(\lambda,\mu)\coloneqq Q({\cal F}_{\lambda,\mu})$, the two-way quantum capacity $Q_2(\lambda,\mu)\coloneqq Q_2({\cal F}_{\lambda,\mu})$, and the secret-key capacity $K(\lambda,\mu)\coloneqq K({\cal F}_{\lambda,\mu})$ read~\cite{mele2023optical}
\bb\label{capacities_formula_diehard3_main}
    Q({\lambda,\mu})=\int_{0}^{2\pi}\frac{\mathrm{d}\theta}{2\pi}Q(\mathcal{E}_{\eta_{\lambda,\mu}(\theta)})\,,\qquad \qquad 
  Q_2({\lambda,\mu}) =
  K({\lambda,\mu}) =\int_{0}^{2\pi}\frac{\mathrm{d}\theta}{2\pi}Q_2(\mathcal{E}_{\eta_{\lambda,\mu}(\theta)})\,,
\ee
where $Q(\mathcal{E}_{\eta_{\lambda,\mu}(\theta)})$ and $Q_2(\mathcal{E}_{\eta_{\lambda,\mu}(\theta)})$ are the capacities reported in Eq.~(\ref{capacities_pure_loss_main}) of the pure loss channel $\mathcal{E}_{\eta_{\lambda,\mu}(\theta)}$ of transmissivity $\eta_{\lambda,\mu}(\theta)$, where
\bb\label{eff_tr_fu}
\eta_{\lambda,\mu}(\theta)\coloneqq \lambda^{\frac{1-\mu}{1-2\sqrt{\mu}\cos(\theta)+\mu}}\qquad\forall\theta\in[0,2\pi]\;.
 \ee
Since these capacities are monotonically increasing in the memory parameter $\mu$~\cite{mele2023optical}, it follows that memory effects improve the communication capabilities of optical fibres. Additionally, by exploiting Eq.~\eqref{capacities_formula_diehard3_main} and Eq.~\eqref{capacities_pure_loss_main},  it easily follows that the quantum capacity $Q({\lambda,\mu})$ vanishes if and only if  $\lambda^{\frac{1-\sqrt \mu}{1+\sqrt\mu}}\le \frac12$. This implies the following remarkable phenomenon~\cite{mele2023optical}: for any arbitrarily small non-zero value of the transmissivity $\lambda$, if the memory parameter $\mu$ is sufficiently large (so that the condition $\lambda^{\frac{1-\sqrt \mu}{1+\sqrt\mu}}>\frac{1}{2}$ is satisfied), then qubit distribution across an optical fibre of transmissivity $\lambda$ becomes achievable. This is remarkable because, in the absence of memory effects, it is impossible to perform qubit distribution across an optical fibre of transmissivity $\le\frac{1}{2}$ (as its quantum capacity vanishes).

\section{Non-asymptotic quantum communication across optical fibres with memory effects}\label{Sec_memory_main}
According to Eq.~\eqref{defmem-ncap}, the $n$-shot capacities of the model of optical fibres with memory effects of Ref.~\cite{mele2023optical} are defined in terms of the one-shot capacities of the $n$-mode Gaussian channel $\Phi^{(n)}_{\lambda,\mu}$ as 
\bb
  Q^{(\varepsilon,n)}({\lambda,\mu}) 
  \coloneqq Q^{(\varepsilon,1)}\left(\Phi^{(n)}_{\lambda,\mu}\right)\,, \qquad 
    Q_2^{(\varepsilon,n)}({\lambda,\mu})\coloneqq Q_2^{(\varepsilon,1)}\left(\Phi^{(n)}_{\lambda,\mu}\right)\,,\qquad    
    K^{(\varepsilon,n)}({\lambda,\mu})\coloneqq K^{(\varepsilon,1)}\left(\Phi^{(n)}_{\lambda,\mu}\right)\,.
\ee
The goal of this paper is to obtain easily computable lower bounds on these quantities. This is important as it allows one to answer the following practical question: How many uses of an optical fibre of transmissivity $\lambda$ and memory parameter $\mu$ are sufficient to transmit $k$ qubits, distil $k$ ebits, or generate $k$ secret-key bits within a given error tolerance $\varepsilon$? We derive such bounds on the non-asymptotic capacities in the following theorem, proved in Theorems~\ref{sm_thm_q_mem}-\ref{sm_thm_k_mem} in the Appendix.
\begin{thm}[(Non-asymptotic quantum communication across an optical fibre with memory effects)]\label{thm_main_memoryq}
For all $n\ge4$ and all $\lambda,\mu,\varepsilon\in(0,1)$, the $n$-shot capacities of the above non-Markovian model of optical fibre with transmissivity $\lambda$ and the memory parameter $\mu$ can be lower bounded as:
\bb\label{ineq_best_bounds_main01112}
    Q^{(\varepsilon,n)}({\lambda,\mu})\ge nQ({\lambda,\mu})-\sqrt{n}C(\lambda,\mu)-\log_2\!\left( \frac{2^{23}(32-\varepsilon)^2}{ (16-\varepsilon)\varepsilon^6}\right)\,,
\ee
\bb\label{ineq_best_bounds_main0133}
    K^{(\varepsilon,n)}({\lambda,\mu})\ge Q_2^{(\varepsilon,n)}({\lambda,\mu})\ge n Q_2({\lambda,\mu})-\sqrt{n}C_2(\lambda,\mu)-\log_2\!\left(\frac{2^{6}\,3\,(4-\sqrt{\varepsilon})^2}{(2-\sqrt{\varepsilon})\varepsilon^3 }\right)\,,
\ee
where $Q({\lambda,\mu})$ and $Q_2({\lambda,\mu})$ are the (asymptotic) capacities of the model reported in Eq.~\eqref{capacities_formula_diehard3_main}, and 
\bb             
    C(\lambda,\mu)&\coloneqq
 \sqrt 8\lambda^{\frac{1-\sqrt \mu}{1+\sqrt\mu}} \sqrt {\mu} \ln\!\left(\frac{1}{\lambda}\right) \frac{1+\sqrt \mu \ln\!\left(\frac{1}{\lambda}\right) +\mu}{(1-\sqrt \mu)^4}+\frac{4(2\pi)^{3/2}(\log_2e)}{1-\lambda^{\frac{1-\sqrt \mu}{1+\sqrt\mu}}}+8\log_2\!\left(\frac{\lambda^{\frac{1-\sqrt \mu}{1+\sqrt\mu}}}{1-\lambda^{\frac{1-\sqrt \mu}{1+\sqrt\mu}}}\right)\,,\\    
C_2(\lambda,\mu)&\coloneqq
 \sqrt 8\lambda^{\frac{1-\sqrt \mu}{1+\sqrt\mu}} \sqrt {\mu} \ln\!\left(\frac{1}{\lambda}\right) \frac{1+\sqrt \mu \ln\!\left(\frac{1}{\lambda}\right) +\mu}{(1-\sqrt \mu)^4}+4(2\pi)^{3/2}(\log_2e)\frac{\lambda^{\frac{1-\sqrt \mu}{1+\sqrt\mu}}}{1-\lambda^{\frac{1-\sqrt \mu}{1+\sqrt\mu}}}+8\log_2\!\left(\frac{1}{1-\lambda^{\frac{1-\sqrt \mu}{1+\sqrt\mu}}}\right)\,.
 \ee
Finally, the lower bound in Eq.~\eqref{ineq_best_bounds_main01112} is trivial if $\lambda^{\frac{1-\sqrt \mu}{1+\sqrt\mu}}\le \frac12$, as the quantum capacity $Q({\lambda,\mu})$ vanishes in such a parameter region.
\end{thm}
The above theorem demonstrates that $n$ uses of an optical fibre with transmissivity $\lambda$ and memory parameter $\mu$ are sufficient to transmit, up to an error $\varepsilon$, the number of qubits specified on the right-hand side of Eq.~\eqref{ineq_best_bounds_main01112} and the number of ebits (and hence of secret-key bits) specified on the right-hand side of Eq.~\eqref{ineq_best_bounds_main0133}.  Hence, Theorem~\ref{thm_main_memoryq} provides a method to determine how many uses of an optical fibre of transmissivity $\lambda$ and memory parameter $\mu$ suffice to transmit a desired number $k$ of qubits, ebits, or secret-key bits up to a given error tolerance $\varepsilon$. Indeed, it suffices to impose that the lower bounds in Eq.~\eqref{ineq_best_bounds_main01112} and Eq.~\eqref{ineq_best_bounds_main0133} are larger than $k$, and solve the corresponding inequalities with respect to $n$.  The simplicity of these lower bounds makes it possible to solve these inequalities by solving a simple quadratic inequality.

Now, let us briefly give the rough idea behind the proof of the above Thereom~\ref{thm_main_memoryq}. In Ref.~\cite{mele2023optical} it has been shown that the $n$-mode Gaussian channel $\Phi^{(n)}_{\lambda,\mu}$ is unitarily equivalent to a tensor product of $n$ pure loss channels, i.e.~there exist unitary channels $\mathcal{U} $ and $\mathcal{V}$ and transmissivities  $\{\eta^{(n,\lambda,\mu)}_i\}_{i=1,2,\ldots,n}$ such that
\bb 
\Phi^{(n)}_{\lambda,\mu}= \mathcal{U}\circ\left(\bigotimes_{i=1}^n\mathcal{E}_{\eta^{(n,\lambda,\mu)}_i}\right)\circ \mathcal{V}\,.
\ee
Consequently, by exploiting the fact that capacities are invariant under unitary transformations, the asymptotic capacities of optical fibres with memory effects are characterised by the asymptotic behaviour for $n\rightarrow\infty$ of the transmissivities $\{\eta^{(n,\lambda,\mu)}_i\}_{i=1,2,\ldots,n}$. Crucially, as shown Ref.~\cite{mele2023optical}, these transmissivities can be retrieved by calculating the singular values of a suitable \emph{Toeplitz matrix}~\cite{Szego1920, GRENADER,Avram1988,Parter1986}. By definition, a Toeplitz matrix is characterised by the property that its element in position $(i,j)$ depends only on the difference $i-j$ for all $i,j$, meaning that all entries along each diagonal are constant. Notably, the asymptotic behaviour of the singular values of a Toeplitz matrix can be analytically understood thanks to the celebrated \emph{Avram--Parter theorem}~\cite{Avram1988,Parter1986}:
\begin{thm}[(\emph{Avram--Parter's theorem}~\cite{Avram1988,Parter1986})]\label{Avram--Partermain}
	Let $\{a_l\}_{l\in\mathbb{Z}}$ be a sequence of real numbers. For all $n\in\N$ let $T^{(n)}$ be the $n\times n$ matrix with elements $T^{(n)}_{k,j}\coloneqq a_{k-j}$ for all $k,j\in\{1,2,\ldots,n\}$. Let $\{s^{(n)}_j\}_{j=1,2,\ldots,n}$ be the singular values of the matrix $T^{(n)}$ ordered in increasing order in $j$. Assume that the function
	$s:[0,2\pi]\to \C$, defined by 
	\bb\label{asymptotic_distribution0main}
	s(x)\coloneqq  \sum_{l=-\infty}^{+\infty}a_l\,e^{ilx} \,\qquad\forall \,x\in[0,2\pi]\,,
	\ee
	is bounded. Then for all continuous functions $F:\mathbb{R}\to\mathbb{R}$ with bounded support it holds that
	\bb
	\lim\limits_{n\rightarrow\infty}\frac{1}{n}\sum_{j=1}^nF\!\left(s^{(n)}_j\right)=\int_{0}^{2\pi}\frac{\mathrm{d}x}{2\pi}F\!\left(|s(x)|\right)\,.
	\ee
\end{thm}
This theorem was exploited in Ref.~\cite{mele2023optical} to show that the  asymptotic behaviour of the transmissivities $\{\eta^{(n,\lambda,\mu)}_i\}_{i=1,2,\ldots,n}$ is determined by the function $\eta_{\lambda,\mu}$ reported in \eqref{eff_tr_fu}, i.e.~it roughly holds that $\eta^{(n,\lambda,\mu)}_{n\frac{\theta}{2\pi}}\xrightarrow{n\rightarrow\infty} \eta_{\lambda,\mu}(\theta)$ for all $\theta\in(0,2\pi]$. This observation was the main tool of Ref.~\cite{mele2023optical} to find the closed expressions of the asymptotic capacities reported in \eqref{capacities_formula_diehard3_main}. Instead, in order to obtain easily computable lower bounds on the $n$-shot capacities of optical fibres with memory effects, the Avram--Parter's theorem alone is insufficient, as it only provides tools for analysing the asymptotic behaviour of Toeplitz matrices (as it inherently regards the asymptotic limit $n\rightarrow\infty$). Therefore, it becomes necessary to establish explicit error bounds for the Avram--Parter's theorem. We do this in the following theorem, which is of independent interest.

\begin{thm}[(Error bound on the Avram--Parter's theorem)]\label{th:final_estimation_lip_main0}
	Let $\{a_l\}_{l\in\mathbb{Z}}$ be a sequence of real numbers. For all $n\in\N$ let $T^{(n)}$ be the $n\times n$ matrix with elements $T^{(n)}_{k,j}\coloneqq a_{k-j}$ for all $k,j\in\{1,2,\ldots,n\}$. Let $\{s^{(n)}_j\}_{j=1,2,\ldots,n}$ be the singular values of the matrix $T^{(n)}$ ordered in increasing order in $j$. Let us define the function
	$s:\mathbb{R}\to \C$ as
	\bb\label{asymptotic_distribution00}
	s(x)\coloneqq \sum_{l=-\infty}^{+\infty}a_l\,e^{ilx}  \qquad\forall \,x\in\mathbb{R}\,.
	\ee
	Let $F:\mathbb{R}\to\mathbb{C}$ be a continuous complex-valued function with bounded support, and assume that $F$ is Lipschitz with Lipschitz constant equal to $L$. Then, for all $n\ge4 $ and all $k\in\mathbb{N}$ such that the $k$-th derivative of $s$ exists and is continuous, it holds that
	\bb\label{key_eq_er_bo}
	\left| \frac{1}{n}\sum_{j=1}^nF\!\left(s_j^{(n)}\right)-\int_{0}^{2\pi}\frac{\mathrm{d}x}{2\pi}F\left(|s(x)|\right) \right|\le  \left(
	\frac {2^{k+1}\|s^{(k)}\|_2L}{\sqrt{2\pi}}  
	+    
	4\pi L \|s\|_2
	\right)
	\frac1{n^{\frac {k}{k+3/2}}}
	+ 
	(2\|F'\|_1 
	+ 
	4\|F\|_\infty)
	\frac{1}{n^{\frac {k+1/2}{k+3/2}}}\,,
	\ee
	where $s^{(k)}$ denotes the $k$-th derivative of $s$, $F'$ denotes the first derivative of $F$, and 
	\bb 
	\|s^{(k)}\|_2&\coloneqq \sqrt{\int_{0}^{2\pi}\mathrm{d}x|s^{(k)}(x)|^2}\,,\\
	\|s\|_2&\coloneqq \sqrt{\int_{0}^{2\pi}\mathrm{d}x|s(x)|^2}\,,\\
	\|F'\|_1&\coloneqq \int_{-\infty}^{\infty}\mathrm{d}x|F'(x)|\,,\\
	\|F\|_\infty&\coloneqq \sup_{x\in\mathbb{R}}|F(x)|\,.
	\ee
\end{thm} 
The proof of the above theorem is provided in Section~\ref{ergodic_estimates} in the Appendix. Note that this theorem is stronger than the Avram--Parter's theorem, as the latter can be recovered by taking the limit of Eq.~(\ref{key_eq_er_bo}) as $n \rightarrow \infty$. This theorem is the crucial tool that we exploit to derive the lower bounds on the $n$-shot capacities of optical fibres with memory effects reported in Theorem~\ref{thm_main_memoryq}. We refer to Section~\ref{memory_chapter} in the Appendix for more details.

\section{Discussion}\label{sec:discussion} 
In this work, we have analysed the performance of non-asymptotic quantum comunication across suitable non-Markovian optical links. In Theorem~\ref{thm_main_memoryq}, we have provided easily computable lower bounds on the non-asymptotic capacities of the model of non-Markovian optical fibre introduced in Ref.~\cite{mele2023optical}. This result allows one to calculate how many uses of a non-Markovian optical fibre are sufficient in order to transmit $k$ qubits, distil $k$ ebits, or generate $k$ secret-key bits up to a fixed error. In order to derive such lower bounds, in Theorem~\ref{th:final_estimation_lip_main0} we introduced an explicit upper bound on the convergence of the celebrated Avram–Parter theorem, which we regard as a tool of independent interest for the field of quantum information theory and matrix analysis.

We leave as an open problem to determine upper bounds on the non-asymptotic capacities of the above model. Moreover, it would be interesting to estimate the non-asymptotic capacities of the above model with the additional presence of thermal noise and amplifier noise.
\medskip

\begin{acknowledgments}
\smallskip
\noindent \emph{Acknowledgements.}  We thank Ludovico Lami and Mark Wilde for useful discussions. F.A.M. and V.G. acknowledge financial support by MUR (Ministero dell'Istruzione, dell'Universit\`a e della Ricerca) through the following projects: PNRR MUR project PE0000023-NQSTI, PRIN 2017 Taming complexity via Quantum Strategies: a Hybrid Integrated Photonic approach (QUSHIP) Id.\ 2017SRN-BRK, and project PRO3 Quantum Pathfinder.
G.B. is member of the Research Group GNCS (Gruppo Nazionale per il Calcolo Scientifico) of INdAM (Istituto Nazionale di Alta Matematica) and  acknowledges the support by the ERC Consolidator Grant 101085607 through the Project eLinoR.  M. F. is supported by the European Research Council (ERC) under Agreement 818761 and by VILLUM FONDEN via the QMATH Centre of Excellence (Grant No. 10059). M.F. was previously supported by a Juan de la Cierva Formaci\'on fellowship (Spanish MCIN project FJC2021-047404-I), with funding from MCIN/AEI/10.13039/501100011033 and European Union NextGenerationEU/PRTR, by European Space Agency, project ESA/ESTEC 2021-01250-ESA, by Spanish MCIN (project PID2022-141283NB-I00) with the support of FEDER funds, by the Spanish MCIN with funding from European Union NextGenerationEU (grant PRTR-C17.I1) and the Generalitat de Catalunya, as well as the Ministry of Economic Affairs and Digital Transformation of the Spanish Government through the QUANTUM ENIA ``Quantum Spain'' project with funds from the European Union through the Recovery, Transformation and Resilience Plan - NextGenerationEU within the framework of the "Digital Spain 2026 Agenda".
\end{acknowledgments}

\medskip


\bibliographystyle{apsrev4-1}
\nocite{apsrev41Control}
\bibliography{biblio,bibarb,revtex-custom}

\newpage
\appendix
\tableofcontents

\section{Lower bounds on the non-asymptotic capacities of non-Markovian bosonic communication lines}\label{memory_chapter}
This section is devoted to the proof of Theorem~\ref{thm_main_memoryq}, which establishes lower bounds on the non-asymptotic capacities of the model of non-Markovian optical fibre of Ref.~\cite{mele2023optical}.

Let us start by recalling some properties and notation regarding this model. It is characterised by two parameters: the transmissivity $\lambda\in[0,1]$ and the memory parameter $\mu\in[0,1]$ of the optical fibre. Such a model describes an optical fibre of transmissivity $\lambda$ and memory parameter $\mu$ via the family of quantum channels
\bb
	 \{ \Phi^{(1)}_{\lambda,\mu},\Phi^{(2)}_{\lambda,\mu}, \cdots, \Phi^{(n)}_{\lambda,\mu}, \cdots\}\;,
\ee
where $\Phi^{(n)}_{\lambda,\mu}$ is an $n$-mode bosonic Gaussian channel that models the evolution the first $n$ input signals. The interferometric representation of $\Phi^{(n)}_{\lambda,\mu}$ is reported in Fig.~\ref{fig_memory_main} in the main text. For any $\varepsilon\in[0,1]$ and $n\in\mathbb{N}$, the $n$-shot quantum capacity, the $n$-shot two-way quantum capacity, and the $n$-shot secret-key capacity of the model are defined in terms of the corresponding one-shot capacities of the quantum channel $\Phi^{(n)}_{\lambda,\mu}$ as
\bb  
    Q^{(\varepsilon,n)}(\lambda,\mu)&\coloneqq Q^{(\varepsilon,1)}\!\left(\Phi^{(n)}_{\lambda,\mu}\right)\,, \\
      Q_2^{(\varepsilon,n)}(\lambda,\mu)&\coloneqq Q_2^{(\varepsilon,1)}\!\left(\Phi^{(n)}_{\lambda,\mu}\right)\,,\\
    K^{(\varepsilon,n)}(\lambda,\mu)&\coloneqq K^{(\varepsilon,1)}\left(\Phi^{(n)}_{\lambda,\mu}\right)\,.
\ee
The goal of this section is to find an easily computable lower bound on these non-asymptotic capacities --- also referred to as $n$-shot capacities --- in terms of the parameters $n,\varepsilon,\lambda,\mu$.

The asymptotic capacities of the model are defined by taking the limit of vanishing error and infinite channel uses of the corresponding $n$-shot capacities. Specifically, the quantum capacity, the two-way quantum capacity, and the secret-key capacity are defined as
\bb
    Q(\lambda,\mu)&\coloneqq
    \lim\limits_{\varepsilon \rightarrow 0^+} \liminf_{n \to \infty} \frac{Q^{(\varepsilon,n)}(\lambda,\mu)}{n}\,,\\
    Q_2(\lambda,\mu)&\coloneqq
    \lim\limits_{\varepsilon \rightarrow 0^+} \liminf_{n \to \infty} \frac{Q_2^{(\varepsilon,n)}(\lambda,\mu)}{n}\,,\\
    K(\lambda,\mu)&\coloneqq
    \lim\limits_{\varepsilon \rightarrow 0^+} \liminf_{n \to \infty} \frac{K^{(\varepsilon,n)}(\lambda,\mu)}{n}\,.
\ee
These asymptotic capacities have been exactly computed in Ref.~\cite{mele2023optical}, and they read:
\bb\label{capacities_formula_diehard3}
    Q(\lambda,\mu)&=\int_{0}^{2\pi}\frac{\mathrm{d}\theta}{2\pi}q(\eta_{\lambda,\mu}(\theta))\,,\\
    Q_2(\lambda,\mu)=K(\lambda,\mu)&=\int_{0}^{2\pi}\frac{\mathrm{d}\theta}{2\pi}k(\eta_{\lambda,\mu}(\theta))\,,
\ee
where $\eta_{\lambda,\mu}:[0,2\pi]\to[0,1]$ is defined as 
\bb\label{def_effect_transm}
    \eta_{\lambda,\mu}(\cdot)\coloneqq \lambda^{\frac{1-\mu}{1-2\sqrt{\mu}\cos(\cdot)+\mu}}\,,
\ee
and $q(\lambda), k(\lambda)$ denote the quantum capacity and the two-way quantum capacity (and hence the secret-key capacity), respectively, of the pure loss channel $\mathcal{E}_\lambda$:
\bb\label{def_q}
    q(\lambda)&\coloneqq Q(\mathcal{E}_{\lambda})=   
        \begin{cases}
        \log_2\!\left(\frac{\lambda}{1-\lambda}\right) &\text{if $\lambda\in(\frac{1}{2},1]$ ,} \\
        0 &\text{if $\lambda\in[0,\frac{1}{2}]$ ,}
    \end{cases}\\
    k(\lambda)&\coloneqq K(\mathcal{E}_{\lambda})=Q_2(\mathcal{E}_{\lambda})= \log_2\!\left(\frac{1}{1-\lambda}\right)\,.
\ee 
In the following lemma we state a useful property regarding the channel $\Phi^{(n)}_{\lambda,\mu}$~\cite{mele2023optical}.
\begin{lemma}[(Factorisation of the quantum channel modelling an optical fibre with memory effects~\cite{mele2023optical})]\label{lemma_unitarily}
Let $\lambda,\mu\in[0,1]$ and $n\in\mathbb{N}$. The channel $\Phi^{(n)}_{\lambda,\mu}$ is unitarily equivalent to a tensor product of $n$ distinct pure loss channels. That is, there exist unitary channels $\mathcal{U}$ and $\mathcal{V}$ such that
\bb\label{memory_channel}
    \Phi^{(n)}_{\lambda,\mu}= \mathcal{U}\circ\left(\bigotimes_{i=1}^n\mathcal{E}_{\eta^{(n,\lambda,\mu)}_i}\right)\circ \mathcal{V}\,,
\ee
where the transmissivities $\{\eta^{(n,\lambda,\mu)}_i\}_{i=1,2,\ldots,n}$ are given by 
\bb\label{trasm_etai}
        \eta_i^{(n,\lambda,\mu)}\coloneqq (s_i^{(n,\lambda,\mu)})^2\,,
\ee
where $s_i^{(n,\lambda,\mu)}$ is the $i$th smallest singular value of the $n\times n$ top-left corner of the infinite matrix $\bar{A}^{(\lambda,\mu)}$, defined as follows. For all $i,k\in\mathbb{N}$, the $(i,k)$ element of the infinite matrix $\bar{A}^{(\lambda,\mu)}$ is given by
        \bb\label{matrix_sm22}
            \bar{A}^{(\lambda,\mu)}_{i,k}\coloneqq a^{(\lambda,\mu)}_{i-k}\,,
        \ee
    where 
    \bb\label{eq_elements}
        a_j^{(\lambda,\mu)}\coloneqq  \Theta(j)\,\sqrt{\lambda}\mu^{\frac{j}{2}} L_{j}^{(-1)}(-\ln\lambda)\qquad \forall\, j\in\mathbb{Z}\,.
    \ee
In \eqref{eq_elements}, $\Theta$ denotes the Heaviside Theta function defined as 
    \bb
\Theta(x)\coloneqq\begin{cases}
1, & \text{if $x\ge0$,} \\
0, & \text{otherwise,}
\end{cases}
    \ee
while $\{L_m^{(-1)}\}_{m\in\N}$ are the generalised Laguerre polynomials defined as 
    \bb
        L_m^{(-1)}(x)&\coloneqq \sum_{l=1}^m\binom{m-1}{l-1}\frac{(-x)^l}{l!}\qquad\forall\,m\in\mathbb{N}^+ \,,\\
        L_0^{(-1)}(x)&\coloneqq 1\,,
    \ee 
for all $x\in\mathbb{R}$. 
\end{lemma}
By exploiting this lemma, together with the simple fact that the capacities are invariant under composition with unitary channels, it follows that the $n$-shot capacities of the model can be expressed in terms of the one-shot capacity of the following tensor product of $n$ pure loss channels:
\bb\label{eqmpde_d}
    Q^{(\varepsilon,n)}(\lambda,\mu)&=Q^{(\varepsilon,1)}\left(  \bigotimes_{i=1}^n\mathcal{E}_{\eta^{(n,\lambda,\mu)}_i} \right)\,,\\
    Q_2^{(\varepsilon,n)}(\lambda,\mu)&=Q_2^{(\varepsilon,1)}\left(  \bigotimes_{i=1}^n\mathcal{E}_{\eta^{(n,\lambda,\mu)}_i} \right)\,,\\
    K^{(\varepsilon,n)}(\lambda,\mu)&=K^{(\varepsilon,1)}\left(  \bigotimes_{i=1}^n\mathcal{E}_{\eta^{(n,\lambda,\mu)}_i} \right)\,,
\ee
where the transmissivities $\{\eta^{(n,\lambda,\mu)}_i\}_{i=1,2,\ldots,n}$ are defined in \eqref{trasm_etai}. In order to find an easily computable lower bound on these one-shot capacities, we exploit a generalisation of the \emph{relative entropy variance approach} developed in Ref.~\cite{mele2025achievableratesnonasymptoticbosonic}. This approach will lead us to the following lemma, which can regarded as a generalisation of \cite[Theorem~5]{mele2025achievableratesnonasymptoticbosonic} to tensor product of pure loss channels.
\begin{lemma}[(Lower bound on the one-shot capacities of a tensor product of pure loss channels via entropy variance approach)]\label{thm_l_rel_non_iid22_sm}
Let $n\in\mathbb{N}$, $\lambda_1,\ldots,\lambda_n\in[0,1]$, and $\varepsilon\in(0,1)$. The one-shot quantum capacity, the one-shot two-way quantum capacity, and the one-shot secret-key capacity of the tensor product of pure loss channels $\bigotimes_{i=1}^n\mathcal{E}_{\lambda_i}$ can be lower bounded as follows:
\bb\label{eq_l_bound_one_sm}
    Q^{(\varepsilon,1)}\left(\bigotimes_{i=1}^n\mathcal{E}_{\lambda_i}\right) 
        &\geq  \sum_{i=1}^n Q(\mathcal{E}_{\lambda_i})-\log_2\!\left( \frac{2^{23}(32-\varepsilon)^2}{ (16-\varepsilon)\varepsilon^6}\right)\,,\\
        K^{(\varepsilon,1)}\left(\bigotimes_{i=1}^n\mathcal{E}_{\lambda_i}\right)&\ge Q_2^{(\varepsilon,1)}\left(\bigotimes_{i=1}^n\mathcal{E}_{\lambda_i}\right) \ge  \sum_{i=1}^n Q_2\left(\mathcal{E}_{\lambda_i}\right)-\log_2\!\left(\frac{2^{6}\,3\,(4-\sqrt{\varepsilon})^2}{(2-\sqrt{\varepsilon})\varepsilon^3 }\right)\,,
    \ee
where $Q(\mathcal{E}_{\lambda})$ and $Q_2(\mathcal{E}_{\lambda})$ are reported in \eqref{def_q}.
\end{lemma}
We provide the proof of this lemma in Section~\ref{sec_low_tens_prodd} below. By exploting this lemma and Eq.~\eqref{eqmpde_d}, we can lower bound the $n$-shot capacities of the model as follows:
\bb\label{eq_non_as_q_q2epsn}
    Q^{(\varepsilon,n)}(\lambda,\mu)&\geq  \sum_{i=1}^n Q\left(\mathcal{E}_{\eta^{(n,\lambda,\mu)}_i}\right)-\log_2\!\left( \frac{2^{23}(32-\varepsilon)^2}{ (16-\varepsilon)\varepsilon^6}\right)\,,\\
        K^{(\varepsilon,n)}(\lambda,\mu)&\ge Q_2^{(\varepsilon,n)}(\lambda,\mu)\ge  \sum_{i=1}^n Q_2\left(\mathcal{E}_{\eta^{(n,\lambda,\mu)}_i}\right)-\log_2\!\left(\frac{2^{6}\,3\,(4-\sqrt{\varepsilon})^2}{(2-\sqrt{\varepsilon})\varepsilon^3 }\right)\,,
\ee
where the transmissivities $\{\eta^{(n,\lambda,\mu)}_i\}_{i=1,2,\ldots,n}$, defined in \eqref{trasm_etai}, can be retrieved by calculating the singular values of the matrix $\bar{A}^{(\lambda,\mu)}$, defined in \eqref{matrix_sm22}, which is a \emph{Toeplitz matrix}~\cite{Szego1920, GRENADER,Avram1988,Parter1986}. By definition, a Toeplitz matrix is characterised by the property that its element in position $(i,j)$ depends only on the difference $i-j$ for all $i,j$, meaning that all entries along each diagonal are constant. Crucially, the asymptotic behaviour of the singular values of a Toeplitz matrix can be analytically understood thanks to the celebrated \emph{Avram--Parter theorem}~\cite{Avram1988,Parter1986}:
\begin{thm}[(\emph{Avram--Parter's theorem}~\cite{Avram1988,Parter1986})]
    Let $\{a_k\}_{k\in\mathbb{Z}}$ be a sequence of real numbers. For all $n\in\N$ let $T^{(n)}$ be the $n\times n$ Toeplitz matrix with elements $T^{(n)}_{k,j}\coloneqq a_{k-j}$ for all $k,j\in\{1,2,\ldots,n\}$. Let $\{s^{(n)}_j\}_{j=1,2,\ldots,n}$ be the singular values of the matrix $T^{(n)}$ ordered in increasing order in $j$. Assume that the function
    $s:[0,2\pi]\to \C$, defined by 
    \bb\label{asymptotic_distribution0}
    s(x)\coloneqq  \sum_{k=-\infty}^{+\infty}a_k\,e^{ikx} \,\qquad\forall \,x\in[0,2\pi]\,,
    \ee
    is bounded. Then for all continuous function $F:\mathbb{R}\to\mathbb{R}$ with bounded support it holds that
    \bb
        \lim\limits_{n\rightarrow\infty}\frac{1}{n}\sum_{j=1}^nF\!\left(s^{(n)}_j\right)=\int_{0}^{2\pi}\frac{\mathrm{d}x}{2\pi}F\!\left(|s(x)|\right)\,.
    \ee
\end{thm}
In \cite{mele2023optical}, the Avram--Parter's theorem was exploited to derive the closed expression of the asymptotic capacities $Q(\lambda,\mu)$, $Q_2(\lambda,\mu)$, and $K(\lambda,\mu)$ reported in \eqref{capacities_formula_diehard3}. More explicitly, in~\cite{mele2023optical}, the Avram--Parter's theorem was exploited to prove that the quantum capacity $Q(\lambda,\mu)$ is given by 
\bb\label{q_memory_die_hard3}
    Q(\lambda,\mu)=\lim\limits_{n\rightarrow\infty}\frac{1}{n}\sum_{i=1}^n Q\left(\mathcal{E}_{\eta^{(n,\lambda,\mu)}_i}\right)=\int_{0}^{2\pi}\frac{\mathrm{d}\theta}{2\pi} Q\left(\mathcal{E}_{\eta_{\lambda,\mu}(\theta)}\right)\,,
\ee
where $\eta_{\lambda,\mu}(\cdot)$ is defined in \eqref{def_effect_transm}. In \eqref{q_memory_die_hard3}, the role of the function $F$ in the statement of the Avram-Parter's theorem is played by the function $\lambda\mapsto Q(\mathcal{E}_\lambda)$, while the role of the function $s$ is played by the function $f_{\lambda,\mu}$ defined as
\bb\label{symbol}
    f_{\lambda,\mu}(\theta)\coloneqq \sum_{k\in\mathbb{Z}} a^{(\lambda,\mu)}_k e^{ik\theta}=\lambda^{-\frac{1}{2}+\frac{1}{1-\sqrt{\mu}\exp({i\theta})}} \qquad\forall\theta\in[0,2\pi]\,,
\ee
which satisfies $|f_{\lambda,\mu}(\cdot)|^2=\eta_{\lambda,\mu}(\cdot)$, with $a^{(\lambda,\mu)}_k$ being defined in \eqref{matrix_sm22}. Additionally, in~\cite{mele2023optical}, the Avram-Parter's theorem was exploited to show that the two-way quantum capacity $Q_2(\lambda,\mu)$ and the secret key capacity $K(\lambda,\mu)$ are both given by
\bb\label{q2_memory_die_hard_3}
    Q_2(\lambda,\mu)=K(\lambda,\mu)=\lim\limits_{n\rightarrow\infty}\frac{1}{n}\sum_{i=1}^n Q_2\left(\mathcal{E}_{\eta^{(n,\lambda,\mu)}_i}\right)=\int_{0}^{2\pi}\frac{\mathrm{d}\theta}{2\pi}Q_2\left(\mathcal{E}_{|f_{\lambda,\mu}(\theta)|^2}\right)\,,
\ee
where now the role of the function $F$ in the statement of the Avram-Parter's theorem is played by the function $\lambda\mapsto Q_2(\mathcal{E}_\lambda)$. While the Avram-Parter's theorem was crucial to calculate the asymptotic capacities~\cite{mele2023optical}, it is not enough to derive a lower bound on the $n$-shot capacities (as the Avram-Parter's inherently regards the asymptotic limit $n\rightarrow\infty$).

Thanks to \eqref{eq_non_as_q_q2epsn}, the $n$-shot capacities can be lower bounded in terms of the sums 
\bb\label{eq_sumssss}
\sum_{i=1}^n Q\!\left(\mathcal{E}_{\eta^{(n,\lambda,\mu)}_i}\right)\quad\text{and}\quad\sum_{i=1}^n Q_2\!\left(\mathcal{E}_{\eta^{(n,\lambda,\mu)}_i}\right)\,.
\ee
These sums are far from being easily computable, as the transmissivities $\{\eta_i^{(n,\lambda,\mu)}\}_{i=1,\ldots,n}$ are related to the singular values of an $n\times n$ matrix, and the dimension $n$ can be very large. However, the Avram--Parter's theorem establishes that these two sums, divided by $n$, approach the integrals  
\bb\label{eq_integralssss}
\int_{0}^{2\pi}\frac{\mathrm{d}\theta}{2\pi}Q\!\left(\mathcal{E}_{\eta_{\lambda,\mu}(\theta)}\right)\quad\text{and}\quad\int_{0}^{2\pi}\frac{\mathrm{d}\theta}{2\pi}Q_2\!\left(\mathcal{E}_{\eta_{\lambda,\mu}(\theta)}\right)
\ee
in the limit $n\rightarrow\infty$. Hence, in order to find an easily computable lower bound on the one-shot capacities, it suffices to find explicit upper bounds on the error incurred in the approximation of the sums in \eqref{eq_sumssss} with the integrals in \eqref{eq_sumssss} provided by the Avram--Parter's theorem. That is, we need to upper bound the errors
\bb\label{quantities_to_bound}
\left|\frac{1}{n}\sum_{i=1}^n Q\!\left(\mathcal{E}_{\eta^{(n,\lambda,\mu)}_i}\right)-\int_{0}^{2\pi}\frac{\mathrm{d}\theta}{2\pi}Q\!\left(\mathcal{E}_{\eta_{\lambda,\mu}(\theta)}\right)\right|\quad\text{and}\quad\left|\frac{1}{n}\sum_{i=1}^n Q_2\!\left(\mathcal{E}_{\eta^{(n,\lambda,\mu)}_i}\right)-\int_{0}^{2\pi}\frac{\mathrm{d}\theta}{2\pi}Q_2\!\left(\mathcal{E}_{\eta_{\lambda,\mu}(\theta)}\right)\right| \,.
\ee
To do this, we establish the following bound on the convergence of the Avram-Parter's theorem, which may be of independent interest for the matrix analysis literature.
\begin{thm}[(Error bound on the Avram--Parter's theorem)]\label{th:final_estimation_lip_main}
Let $\{a_l\}_{l\in\mathbb{Z}}$ be a sequence of real numbers. For all $n\in\N$ let $T^{(n)}$ be the $n\times n$ matrix with elements $T^{(n)}_{k,j}\coloneqq a_{k-j}$ for all $k,j\in\{1,2,\ldots,n\}$. Let $\{s^{(n)}_j\}_{j=1,2,\ldots,n}$ be the singular values of the matrix $T^{(n)}$ ordered in increasing order in $j$. Let us define the function
    $s:\mathbb{R}\to \C$ as
    \bb\label{asymptotic_distribution001}
    s(x)\coloneqq \sum_{l=-\infty}^{+\infty}a_l\,e^{ilx}  \qquad\forall \,x\in\mathbb{R}\,.
    \ee
     Let $F:\mathbb{R}\to\mathbb{C}$ be a continuous complex-valued function with bounded support, and assume that $F$ is Lipschitz with Lipschitz constant equal to $L$. Then, for all $n\ge4 $ and all $k\in\mathbb{N}$ such that the $k$-th derivative of $s$ exists and is continuous, it holds that
    \bb\label{key_eq_er_boSM}
        \left| \frac{1}{n}\sum_{j=1}^nF\!\left(s_j^{(n)}\right)-\int_{0}^{2\pi}\frac{\mathrm{d}x}{2\pi}F\left(|s(x)|\right) \right|\le  \left(
\frac {2^{k+1}\|s^{(k)}\|_2L}{\sqrt{2\pi}}  
+    
4\pi L \|s\|_2
\right)
\frac1{n^{\frac {k}{k+3/2}}}
+ 
(2\|F'\|_1 
 + 
4\|F\|_\infty)
\frac{1}{n^{\frac {k+1/2}{k+3/2}}}\,,
    \ee
where $s^{(k)}$ denotes the $k$-th derivative of $s$, $F'$ denotes the first derivative of $F$, and 
\bb 
    \|s^{(k)}\|_2&\coloneqq \sqrt{\int_{0}^{2\pi}\mathrm{d}x|s^{(k)}(x)|^2}\,,\\
    \|s\|_2&\coloneqq \sqrt{\int_{0}^{2\pi}\mathrm{d}x|s(x)|^2}\,,\\
    \|F'\|_1&\coloneqq \int_{-\infty}^{\infty}\mathrm{d}x|F'(x)|\,,\\
    \|F\|_\infty&\coloneqq \sup_{x\in\mathbb{R}}|F(x)|\,.
\ee
\end{thm} 
The proof of the above theorem is provided in Theorem~\ref{th:final_estimation_lip} in Section~\ref{ergodic_estimates}. Note that the above theorem is stronger than the Avram--Parter's theorem, as the latter can be recovered by taking the limit of Eq.~(\ref{key_eq_er_boSM}) as $n \rightarrow \infty$. We will exploit this theorem in Subsection~\ref{sec__mem_q} to bound the $n$-shot quantum capacity and in Section~\ref{sec__mem_k} to bound the $n$-shot two-way quantum capacity and the $n$-shot secret-key capacity.

\subsection{$n$-shot quantum capacity}\label{sec__mem_q}
This subsection is devoted to show an easily computable lower bound the $n$-shot quantum capacity $Q^{\varepsilon,n}(\lambda,\mu)$. Thanks to~\eqref{eq_non_as_q_q2epsn}, $Q^{\varepsilon,n}(\lambda,\mu)$ can be lower bounded as
\bb\label{eq_lowerq_mem_sm}
    Q^{(\varepsilon,n)}(\lambda,\mu)&\geq  \sum_{i=1}^n Q\left(\mathcal{E}_{\eta^{(n,\lambda,\mu)}_i}\right)-\log_2\!\left( \frac{2^{23}(32-\varepsilon)^2}{ (16-\varepsilon)\varepsilon^6}\right)\,,
\ee
and thus it suffices to find an easily computable lower bound on the quantity $\sum_{i=1}^n Q\left(\mathcal{E}_{\eta^{(n,\lambda,\mu)}_i}\right)$. We do this in the forthcoming Lemma~\ref{lemma_q_memory}.

\begin{lemma}\label{lemma_q_memory}
    For all $\lambda,\mu\in(0,1)$ and all $n\ge 4$, it holds that:
    \bb\label{result_lemma}
    \sum_{i=1}^n Q\!\left(\mathcal{E}_{\eta^{(n,\lambda,\mu)}_i}\right)\ge nQ(\lambda,\mu)-\sqrt{n}C(\lambda,\mu)\,,
\ee
where 
\bb      
C(\lambda,\mu)&\coloneqq
 \sqrt 8M_{\lambda,\mu} \sqrt {\mu} \ln\!\left(\frac{1}{\lambda}\right) \frac{1+\sqrt \mu \ln\!\left(\frac{1}{\lambda}\right) +\mu}{(1-\sqrt \mu)^4}+\frac{4(2\pi)^{3/2}(\log_2e)}{1-\lambda^{\frac{1-\sqrt \mu}{1+\sqrt\mu}}}+8\log_2\!\left(\frac{\lambda^{\frac{1-\sqrt \mu}{1+\sqrt\mu}}}{1-\lambda^{\frac{1-\sqrt \mu}{1+\sqrt\mu}}}\right)\,.
\ee
Moreover, the lower bound is trivial if $\lambda^{\frac{1-\sqrt \mu}{1+\sqrt\mu}}\le \frac12$, as the quantum capacity $Q({\lambda,\mu})$ vanishes in such a parameter region.
\end{lemma}
\begin{proof}
To simplify the notation, let us introduce the function $\eta\mapsto q(\eta)$ and the quantity $M_{\lambda,\mu}$ defined as
\bb
    q(\eta)&\coloneqq Q(\mathcal{E}_\eta)\,,\\
     M_{\lambda,\mu}&\coloneqq \lambda^{\frac{1-\sqrt \mu}{1+\sqrt\mu}}\,.
\ee
By exploiting \eqref{capacities_formula_diehard3} together with the fact that $\max_{\theta\in[0,2\pi]}\eta_{\lambda,\mu}(\theta)=M_{\lambda,\mu}$, it is simple to observe that $Q(\lambda,\mu)>0$ if and only if $M_{\lambda,\mu} >\frac{1}{2}$. Hence, if $M_{\lambda,\mu}\le\frac{1}{2}$, the inequality in \eqref{result_lemma} is trivial.

Now, let us assume that $\lambda$ and $\mu$ are such that $M_{\lambda,\mu}>\frac{1}{2}$.  By exploiting Theorem~\ref{th:norm_toeplitz_symbol} in Section~\ref{ergodic_estimates} with $p=\infty$, we deduce that the largest transmissivity $\eta_n^{(n,\lambda,\mu)}$ can be upper bounded as
\bb\label{max_trasm}
    \eta_n^{(n,\lambda,\mu)}=\left(s_n^{(n,\lambda,\mu)}\right)^2\le \max_{\theta\in[0,2\pi]} |f_{\lambda,\mu}(\theta)|^2=\max_{\theta\in[0,2\pi]} \eta_{\lambda,\mu}(\theta)=M_{\lambda,\mu}\,,
\ee
Moreover, note that since it always holds that $M_{\lambda,\mu}<1$, the quantity $q(M_{\lambda,\mu})$ is well defined. Hence, we can define a function $F:\mathbb R^+\to \mathbb{R}$ as
\bb
    F(x) \coloneqq
    \begin{cases}
    q(x^2), & x\le \sqrt{M_{\lambda,\mu}},\\
    q(M_{\lambda,\mu}) - L(x-\sqrt{ M_{\lambda,\mu}})\,, & \sqrt{ M_{\lambda,\mu}} \le x \le \sqrt{M_{\lambda,\mu}} + \frac{q(M_{\lambda,\mu})}{L}\,,  \\
    0, & \text{otherwise},
    \end{cases}
\ee
where
\bb
    L\coloneqq \frac{2(\log_2 e)}{\sqrt{M_{\lambda,\mu}}(1-M_{\lambda,\mu})}\,.
\ee
By explicitly calculating the derivative $F'$, it can be easily shown that the Lipschitz constant of $F$ is given by $L$. Moreover, as a consequence of \eqref{max_trasm}, $F$ satisfies
\bb
    q\!\left(\eta_i^{(n,\lambda,\mu)}\right)=F\!\left(s_i^{(n,\lambda,\mu)}\right) 
\ee
and
\bb
    q\!\left(|f_{\lambda,\mu}(\theta)|^2\right)&=F\!\left(|f_{\lambda,\mu}(\theta)|\right)
\ee
for every $i, n$, and $\theta$. Hence, we have that
\bb
    \left\lvert \frac 1n\sum_{i=1}^n q\!\left(\eta_i^{(n,\lambda,\mu)}\right) - 
    \int_{0}^{2\pi}\frac{\mathrm{d}\theta}{2\pi}q(|f_{\lambda,\mu}(\theta)|^2)\right\rvert =     \left\lvert \frac 1n\sum_{i=1}^n F\!\left(s_i^{(n,\lambda,\mu)}\right) -\int_{0}^{2\pi}\frac{\mathrm{d}\theta}{2\pi}F(|f_{\lambda,\mu}(\theta)|)\right\rvert \,.
\ee
Hence, since $F$ is continuous, all the hypotheses of Theorem~\ref{th:final_estimation_lip_main} are satisfied. We can thus apply Theorem~\ref{th:final_estimation_lip_main} with $k=2$ and obtain:
\bb
\left\lvert \frac 1n\sum_{i=1}^n q\!\left(\eta_i^{(n,\lambda,\mu)}\right) - 
\int_{0}^{2\pi}\frac{\mathrm{d}\theta}{2\pi}q(|f_{\lambda,\mu}(\theta)|^2)
\right\rvert\le 
 \left(
\frac {8\|f_{\lambda,\mu}''\|_2L}{\sqrt{2\pi}}  
+    
4\pi L \|f_{\lambda,\mu}\|_2
\right)
\frac1{n^{4/7}}
+ 
(2\|F'\|_1 
 + 
4\|F\|_\infty)
\frac{1}{n^{5/7}}\,.
\ee
Hence, it suffices to upper bound the quantities $\|f_{\lambda,\mu}\|_2$, $\|f_{\lambda,\mu}''\|_2$, $\|F\|_\infty$, and $\|F'\|_1$. By an explicit, simple calculation, it follows that
\bb
    \|f_{\lambda,\mu}\|_2&\le \sqrt{2\pi}\max_{\theta\in[0,2\pi]}|f_{\lambda,\mu}|=\sqrt{2\pi M_{\lambda,\mu}}\,,\\
    \|f_{\lambda,\mu}''\|_2&\le \sqrt M_{\lambda,\mu} \sqrt {2\pi\mu} \ln\!\left(\frac{1}{\lambda}\right) \frac{1+\sqrt \mu \ln\!\left(\frac{1}{\lambda}\right) +\mu}{(1-\sqrt \mu)^4}\,,
\ee
and that
\bb
    \|F\|_\infty&=\log_2\!\left(\frac{M_{\lambda,\mu}}{1-M_{\lambda,\mu}}\right)\,,\\
    \|F'\|_1&=\int_{0}^{\sqrt{M_{\lambda,\mu}}} \!\mathrm{d}\!x\,\partial_x [q(x^2)]  + q(M_{\lambda,\mu}) =2\log_2\!\left(\frac{M_{\lambda,\mu}}{1-M_{\lambda,\mu}}\right)\,.
\ee
Consequently, we obtain that
\bb
    \left\lvert \frac 1n\sum_{i=1}^n q\!\left(\eta_i^{(n,\lambda,\mu)}\right) - 
\int_{0}^{2\pi}\frac{\mathrm{d}\theta}{2\pi}q(|f_{\lambda,\mu}(\theta)|^2)
\right\rvert&\le 
\frac{  
\frac {8\|f_{\lambda,\mu}''\|_2L}{\sqrt{2\pi}}  +   
4\pi L \|f_{\lambda,\mu}\|_2
 +  2\|F'\|_1 
 +  4\|F\|_\infty }{\sqrt{n}}\\
&\le\frac{C(\lambda,\mu)}{\sqrt{n}}\,,
\ee
where 
\bb
C(\lambda,\mu)&\coloneqq
 \sqrt 8M_{\lambda,\mu} \sqrt {\mu} \ln\!\left(\frac{1}{\lambda}\right) \frac{1+\sqrt \mu \ln\!\left(\frac{1}{\lambda}\right) +\mu}{(1-\sqrt \mu)^4}+\frac{4(2\pi)^{3/2}(\log_2e)}{1-M_{\lambda,\mu}}+8\log_2\!\left(\frac{M_{\lambda,\mu}}{1-M_{\lambda,\mu}}\right)\,.
\ee
Hence, we conclude that
\bb
    \sum_{i=1}^n q\!\left(\eta_i^{(n,\lambda,\mu)}\right)  &\ge 
n\int_{0}^{2\pi}\frac{\mathrm{d}\theta}{2\pi}q(|f_{\lambda,\mu}(\theta)|^2)-\sqrt{n}C(\lambda,\mu) \\
&= nQ(\lambda,\mu)-\sqrt{n}C(\lambda,\mu)\,,
\ee
where in the last line we exploited \ref{q_memory_die_hard3}.
\end{proof}
By combining \eqref{eq_lowerq_mem_sm} with Lemma~\ref{lemma_q_memory}, we obtain the following.
\begin{thm}\label{sm_thm_q_mem}
For all $n\ge4$ and all $\lambda,\mu,\varepsilon\in(0,1)$, the $n$-shot quantum capacity of the model can be lower bounded as:
\bb
    Q^{(\varepsilon,n)}({\lambda,\mu})\ge nQ({\lambda,\mu})-\sqrt{n}C(\lambda,\mu)-\log_2\!\left( \frac{2^{23}(32-\varepsilon)^2}{ (16-\varepsilon)\varepsilon^6}\right)\,,
\ee
where $Q({\lambda,\mu})$ is the quantum capacity reported in Eq.~\eqref{capacities_formula_diehard3}, and 
\bb             
    C(\lambda,\mu)&\coloneqq
 \sqrt 8\lambda^{\frac{1-\sqrt \mu}{1+\sqrt\mu}} \sqrt {\mu} \ln\!\left(\frac{1}{\lambda}\right) \frac{1+\sqrt \mu \ln\!\left(\frac{1}{\lambda}\right) +\mu}{(1-\sqrt \mu)^4}+\frac{4(2\pi)^{3/2}(\log_2e)}{1-\lambda^{\frac{1-\sqrt \mu}{1+\sqrt\mu}}}+8\log_2\!\left(\frac{\lambda^{\frac{1-\sqrt \mu}{1+\sqrt\mu}}}{1-\lambda^{\frac{1-\sqrt \mu}{1+\sqrt\mu}}}\right)\,,\\
 \ee
Finally, the lower bound is trivial if $\lambda^{\frac{1-\sqrt \mu}{1+\sqrt\mu}}\le \frac12$, as the quantum capacity $Q({\lambda,\mu})$ vanishes in such a parameter region.
\end{thm}

\subsection{$n$-shot two-way quantum capacity and $n$-shot secret-key capacity}\label{sec__mem_k}
This subsection is devoted to show an easily computable lower bound the $n$-shot two-way quantum capacity $Q_2^{\varepsilon,n}(\lambda,\mu)$ and on the $n$-shot secret-key capacity $K^{\varepsilon,n}(\lambda,\mu)$. Thanks to~\eqref{eq_non_as_q_q2epsn}, these $n$-shot capacities can be lower bounded as
\bb\label{eq_lowerq2_mem_sm}
    K^{(\varepsilon,n)}(\lambda,\mu)&\ge Q_2^{(\varepsilon,n)}(\lambda,\mu)\ge  \sum_{i=1}^n Q_2\left(\mathcal{E}_{\eta^{(n,\lambda,\mu)}_i}\right)-\log_2\!\left(\frac{2^{6}\,3\,(4-\sqrt{\varepsilon})^2}{(2-\sqrt{\varepsilon})\varepsilon^3 }\right)\,,
\ee
and thus it suffices to find an easily computable lower bound on the quantity $\sum_{i=1}^n Q_2\left(\mathcal{E}_{\eta^{(n,\lambda,\mu)}_i}\right)$. We do this in the forthcoming Lemma~\ref{lemma_k_memory}.
\begin{lemma}\label{lemma_k_memory}
For all $\lambda,\mu\in(0,1)$ and all $n\ge 4$, it holds that:
    \bb\label{result_lemma2}
    \sum_{i=1}^n Q_2\!\left(\mathcal{E}_{\eta^{(n,\lambda,\mu)}_i}\right)\ge nQ_2(\lambda,\mu)-\sqrt{n}C(\lambda,\mu)\,,
\ee
where 
\bb      
C(\lambda,\mu)&\coloneqq  \sqrt 8\lambda^{\frac{1-\sqrt \mu}{1+\sqrt\mu}} \sqrt {\mu} \ln\!\left(\frac{1}{\lambda}\right) \frac{1+\sqrt \mu \ln\!\left(\frac{1}{\lambda}\right) +\mu}{(1-\sqrt \mu)^4}+4(2\pi)^{3/2}(\log_2e)\frac{\lambda^{\frac{1-\sqrt \mu}{1+\sqrt\mu}}}{1-\lambda^{\frac{1-\sqrt \mu}{1+\sqrt\mu}}}+8\log_2\!\left(\frac{1}{1-\lambda^{\frac{1-\sqrt \mu}{1+\sqrt\mu}}}\right)\,.
\ee
\end{lemma}
\begin{proof}
The proof is similar to the one of Lemma~\ref{lemma_q_memory} and the main technical tool is to employ Theorem~\ref{th:final_estimation_lip_main}.  To simplify the notation, let us introduce the function $\eta\mapsto k(\eta)$ and the quantity $M_{\lambda,\mu}$ defined as
\bb
    k(\eta)&\coloneqq Q_2(\mathcal{E}_\eta)\,,\\
     M_{\lambda,\mu}&\coloneqq \lambda^{\frac{1-\sqrt \mu}{1+\sqrt\mu}}\,.
\ee
As we saw in the proof of Lemma~\ref{lemma_q_memory}, it holds that
\bb\label{max_trasm2}
    \eta_n^{(n,\lambda,\mu)}&\le M_{\lambda,\mu}\,,\\
     \max_{\theta\in[0,2\pi]} |f_{\lambda,\mu}(\theta)|^2&\le M_{\lambda,\mu}\,.
\ee
Let us define $F:\mathbb R^+\to \mathbb{R}$ as
\bb
    F(x) \coloneqq
    \begin{cases}
    k(x^2), & x\le \sqrt{M_{\lambda,\mu}},\\
    k(M_{\lambda,\mu}) - L(x-\sqrt{ M_{\lambda,\mu}})\,, & \sqrt{ M_{\lambda,\mu}} \le x \le \sqrt{M_{\lambda,\mu}} + \frac{k(M_{\lambda,\mu})}{L}\,,  \\
    0, & \text{otherwise},
    \end{cases}
\ee
where
\bb
    L\coloneqq \frac{2(\log_2 e)\sqrt{M_{\lambda,\mu}}}{1-M_{\lambda,\mu}}\,.
\ee
By evaluating the derivative of $F$, it can be easily shown that the Lipschitz constant of $F$ is given by $L$. Moreover, as a consequence of \eqref{max_trasm2}, $F$ satisfies
\bb
    k\!\left(\eta_i^{(n,\lambda,\mu)}\right)=F\!\left(s_i^{(n,\lambda,\mu)}\right) 
\ee
and
\bb
    k\!\left(|f_{\lambda,\mu}(\theta)|^2\right)&=F\!\left(|f_{\lambda,\mu}(\theta)|\right)
\ee
for every $i, n$, and $\theta$. Hence, we have that
\bb
    \left\lvert \frac 1n\sum_{i=1}^n k\!\left(\eta_i^{(n,\lambda,\mu)}\right) - 
    \int_{0}^{2\pi}\frac{\mathrm{d}\theta}{2\pi}k(|f_{\lambda,\mu}(\theta)|^2)\right\rvert =     \left\lvert \frac 1n\sum_{i=1}^n F\!\left(s_i^{(n,\lambda,\mu)}\right) -\int_{0}^{2\pi}\frac{\mathrm{d}\theta}{2\pi}F(|f_{\lambda,\mu}(\theta)|)\right\rvert \,.
\ee
Hence, since $F$ is continuous, all the hypotheses of Theorem~\ref{th:final_estimation_lip_main} are satisfied. We can thus apply Theorem~\ref{th:final_estimation_lip_main} with $k=2$ and obtain:
\bb
\left\lvert \frac 1n\sum_{i=1}^n k\!\left(\eta_i^{(n,\lambda,\mu)}\right) - 
\int_{0}^{2\pi}\frac{\mathrm{d}\theta}{2\pi}k(|f_{\lambda,\mu}(\theta)|^2)
\right\rvert\le 
 \left(
\frac {8\|f_{\lambda,\mu}''\|_2L}{\sqrt{2\pi}}  
+    
4\pi L \|f_{\lambda,\mu}\|_2
\right)
\frac1{n^{4/7}}
+ 
(2\|F'\|_1 
 + 
4\|F\|_\infty)
\frac{1}{n^{5/7}}\,.
\ee
Hence, it suffices to upper bound the quantities $\|f_{\lambda,\mu}\|_2$, $\|f_{\lambda,\mu}''\|_2$, $\|F\|_\infty$, and $\|F'\|_1$. In the proof of Lemma~\ref{lemma_q_memory}, we proved that
\bb
    \|f_{\lambda,\mu}\|_2&\le\sqrt{2\pi M_{\lambda,\mu}}\,,\\
    \|f_{\lambda,\mu}''\|_2&\le \sqrt M_{\lambda,\mu} \sqrt {2\pi\mu} \ln\!\left(\frac{1}{\lambda}\right) \frac{1+\sqrt \mu \ln\!\left(\frac{1}{\lambda}\right) +\mu}{(1-\sqrt \mu)^4}\,.
\ee
Moreover, it holds that
\bb
    \|F\|_\infty&=\log_2\!\left(\frac{1}{1-M_{\lambda,\mu}}\right)\,,\\
    \|F'\|_1&=\int_{0}^{\sqrt{M_{\lambda,\mu}}} \!\mathrm{d}\!x\,\partial_x [k(x^2)]  + k(M_{\lambda,\mu}) =2\log_2\!\left(\frac{1}{1-M_{\lambda,\mu}}\right)\,.
\ee
Consequently, we obtain that
\bb
    \left\lvert \frac 1n\sum_{i=1}^n k\!\left(\eta_i^{(n,\lambda,\mu)}\right) - 
\int_{0}^{2\pi}\frac{\mathrm{d}\theta}{2\pi}k(|f_{\lambda,\mu}(\theta)|^2)
\right\rvert&\le 
\frac{  
\frac {8\|f_{\lambda,\mu}''\|_2L}{\sqrt{2\pi}}  +   
4\pi L \|f_{\lambda,\mu}\|_2
 +  2\|F'\|_1 
 +  4\|F\|_\infty }{\sqrt{n}}\\
&\le\frac{C(\lambda,\mu)}{\sqrt{n}}\,,
\ee
where 
\bb
C(\lambda,\mu)&\coloneqq
 \sqrt 8M_{\lambda,\mu} \sqrt {\mu} \ln\!\left(\frac{1}{\lambda}\right) \frac{1+\sqrt \mu \ln\!\left(\frac{1}{\lambda}\right) +\mu}{(1-\sqrt \mu)^4}+4(2\pi)^{3/2}(\log_2e)\frac{M_{\lambda,\mu}}{1-M_{\lambda,\mu}}+8\log_2\!\left(\frac{1}{1-M_{\lambda,\mu}}\right)\,.
\ee
Hence, we conclude that
\bb
    \sum_{i=1}^n k\!\left(\eta_i^{(n,\lambda,\mu)}\right)  &\ge 
n\int_{0}^{2\pi}\frac{\mathrm{d}\theta}{2\pi}k(|f_{\lambda,\mu}(\theta)|^2)-\sqrt{n}C(\lambda,\mu) \\
&= nK(\lambda,\mu)-\sqrt{n}C(\lambda,\mu)\,,
\ee
where in the last line we exploited \ref{q2_memory_die_hard_3}.
\end{proof}
By combining \eqref{eq_lowerq2_mem_sm} with Lemma~\ref{lemma_k_memory}, we obtain the following.
\begin{thm}\label{sm_thm_k_mem}
For all $n\ge4$ and all $\lambda,\mu,\varepsilon\in(0,1)$, the $n$-shot two-way quantum capacity and the $n$-shot secret-key capacity of the model can be lower bounded as:
\bb
    K^{(\varepsilon,n)}({\lambda,\mu})\ge Q_2^{(\varepsilon,n)}({\lambda,\mu})\ge nQ_2({\lambda,\mu})-\sqrt{n}C(\lambda,\mu)-\log_2\!\left(\frac{2^{6}\,3\,(4-\sqrt{\varepsilon})^2}{(2-\sqrt{\varepsilon})\varepsilon^3 }\right)\,,
\ee
where $Q_2({\lambda,\mu})$ is the two-way quantum capacity reported in Eq.~\eqref{capacities_formula_diehard3}, and 
\bb
    C(\lambda,\mu)&\coloneqq  \sqrt 8\lambda^{\frac{1-\sqrt \mu}{1+\sqrt\mu}} \sqrt {\mu} \ln\!\left(\frac{1}{\lambda}\right) \frac{1+\sqrt \mu \ln\!\left(\frac{1}{\lambda}\right) +\mu}{(1-\sqrt \mu)^4}+4(2\pi)^{3/2}(\log_2e)\frac{\lambda^{\frac{1-\sqrt \mu}{1+\sqrt\mu}}}{1-\lambda^{\frac{1-\sqrt \mu}{1+\sqrt\mu}}}+8\log_2\!\left(\frac{1}{1-\lambda^{\frac{1-\sqrt \mu}{1+\sqrt\mu}}}\right)\,.
\ee
\end{thm}

\newpage

\section{Lower bounding the one-shot capacities of tensor products of Gaussian channels}\label{sec_low_tens_prodd}
In this section, we show how to find an easily computable lower bound the one-shot capacities of tensor product of Gaussian channels. In Subsection~\ref{sec_tens_pure_los}, we analyse the particular case of the tensor product of pure loss channels, by generalising the \emph{entropy variance approach} developed in~\cite{mele2025achievableratesnonasymptoticbosonic}. In Subsection~\ref{app_sec_aep_gen}, we analyse the general case of the tensor product of arbitrary Gaussian channels, by generalising the \emph{asymptotic equipartition property approach} developed in~\cite{mele2025achievableratesnonasymptoticbosonic}.

\subsection{Tensor product of pure loss channels}\label{sec_tens_pure_los}
In order to find an easily computable lower bound on the one-shot capacities of tensor product of pure loss channels, we are going to exploit the \emph{relative entropy variance approach} developed in Ref.~\cite{mele2025achievableratesnonasymptoticbosonic}. This approach will lead us to the following lemma, which can regarded as a generalisation of \cite[Theorem~5]{mele2025achievableratesnonasymptoticbosonic}.
\begin{lemma}[(Lower bound on the one-shot capacities of a tensor product of pure loss channels via entropy variance approach)]\label{thm_l_rel_non_iid22}
Let $n\in\mathbb{N}$, $\lambda_1,\ldots,\lambda_n\in[0,1]$, and $\varepsilon\in(0,1)$. The one-shot quantum capacity, the one-shot two-way quantum capacity, and the one-shot secret-key capacity of the tensor product of pure loss channels $\bigotimes_{i=1}^n\mathcal{E}_{\lambda_i}$ can be lower bounded as follows:
\bb\label{eq_l_bound_one}
    Q^{(\varepsilon,1)}\left(\bigotimes_{i=1}^n\mathcal{E}_{\lambda_i}\right) 
        &\geq  \sum_{i=1}^n Q(\mathcal{E}_{\lambda_i})-\log_2\!\left( \frac{2^{23}(32-\varepsilon)^2}{ (16-\varepsilon)\varepsilon^6}\right)\,,\\
        K^{(\varepsilon,1)}\left(\bigotimes_{i=1}^n\mathcal{E}_{\lambda_i}\right)&\ge Q_2^{(\varepsilon,1)}\left(\bigotimes_{i=1}^n\mathcal{E}_{\lambda_i}\right) \ge  \sum_{i=1}^n Q_2\left(\mathcal{E}_{\lambda_i}\right)-\log_2\!\left(\frac{2^{6}\,3\,(4-\sqrt{\varepsilon})^2}{(2-\sqrt{\varepsilon})\varepsilon^3 }\right)\,,
    \ee
where $Q(\mathcal{E}_{\lambda})$ and $Q_2(\mathcal{E}_{\lambda})$ are reported in \eqref{def_q}.
\end{lemma}
Before proving this lemma, we need to introduce the following preliminary result, which follows by setting $n=1$ in the statement of \cite[Lemma~55]{mele2025achievableratesnonasymptoticbosonic} (with the latter being a trivial consequence of \cite[Proposition 31]{MMMM}). 
\begin{lemma}[(Upper bound for the hypothesis testing relative entropy)]\label{thm_low_pure_loss_non_iid-0}
    Let $\rho$ be positive definite states and $\sigma$ be positive definite operators with bounded trace on a separable Hilbert space. For any $\varepsilon\in(0,1)$ we have that
\bb\label{eq:entropies}
D^{\varepsilon}_{h}(\rho\|\sigma)&\leq D(\rho\|\sigma)+\sqrt{\frac{{V(\rho\|\sigma)}}{1-\varepsilon}}+\log_2 6+ 2\log_2\frac{1+\varepsilon}{1-\varepsilon}\,.
\ee
Here, $D^{\varepsilon}_{h}(\rho\|\sigma)$, $D(\rho\|\sigma)$, and $V(\rho\|\sigma)$ denote the hypothesis testing relative entropy, the quantum relative entropy, and the relative entropy variance, respectively, defined as follows:
\bb
    D^{\varepsilon}_{h}(\rho\|\sigma)&\coloneqq-\log_2\inf_{E}\left\{ \Tr[E\sigma]\,:\, 0\le E\le \mathbb{1},\,\Tr[E\rho]\ge 1-\varepsilon \right\}\,,\\
    D(\rho\|\sigma)&\coloneqq\Tr\left[\rho\left(\log_2\rho-\log_2\sigma\right)\right]\,,\\
V(\rho\|\sigma)&\coloneqq \Tr\left[\rho\left(\log_2 \rho-\log_2\sigma\right)^2\right]-D(\rho\|\sigma)^2\,.
\ee
\end{lemma}
This directly implies the following generalisation of \cite[Lemma~55]{mele2025achievableratesnonasymptoticbosonic} to tensor product of operators.
\begin{lemma}[(Upper bound for the hypothesis testing relative entropy for product operators)]\label{thm_low_pure_loss_non_iid}
    Let $n\in\mathbb{N}$. Let $\rho_1,\ldots\,\rho_n$ be positive definite states and $\sigma_1,\ldots\,\sigma_n$ be positive definite operators with bounded trace on a separable Hilbert space. For any $\varepsilon\in(0,1)$ we have that
\bb
D^{\varepsilon}_{h}\left(\bigotimes_{i=1}^n\rho_i\bigg\|\bigotimes_{i=1}^n\sigma_i\right)&\leq \sum_{i=1}^n D(\rho_i\|\sigma_i)+\sqrt{\frac{{\sum_{i=1}^nV(\rho_i\|\sigma_i)}}{1-\varepsilon}}+\log_2 6+ 2\log_2\frac{1+\varepsilon}{1-\varepsilon}\,,
\ee
where $D^{\varepsilon}_{h}(\rho\|\sigma)$, $D(\rho\|\sigma)$, and $V(\rho\|\sigma)$ are defined in \eqref{eq:entropies}.
\end{lemma}
  \begin{proof}
      The proof follows by using Lemma~\ref{thm_low_pure_loss_non_iid-0} with $\rho=\bigotimes_{i=1}^n\rho_i$ and $\sigma=\bigotimes_{i=1}^n\sigma_i$, and by exploiting the simple-to-prove facts that the relative entropy and the relative entropy variances are additive under tensor product, i.e.
      \bb
        D\left(\bigotimes_{i=1}^n\rho_i\bigg\|\bigotimes_{i=1}^n\sigma_i\right)&=\sum_{i=1}^n D(\rho_i\|\sigma_i)\,,\\
        V\left(\bigotimes_{i=1}^n\rho_i\bigg\|\bigotimes_{i=1}^n\sigma_i\right)&=\sum_{i=1}^n V(\rho_i\|\sigma_i)\,.
      \ee
  \end{proof}
With this lemma at hand, we can prove Lemma~\ref{thm_l_rel_non_iid22}.
\begin{proof}[Proof of Lemma~\ref{thm_l_rel_non_iid22}]
    The statement is a generalisation of \cite[Theorem~5]{mele2025achievableratesnonasymptoticbosonic} to tensor product of distinct pure loss channels. The strategy in \cite{mele2025achievableratesnonasymptoticbosonic} was to lower bound on the achievable rate for $n$ uses of the channel  by a finite-dimensional truncation of the channel and evaluating the bounds available in the finite-dimensional case, in terms of the conditional entropy and the conditional entropy variance, using as input $n$ copies of a truncated version of the two-mode squeezed vacuum. By additivity of the conditional entropy and the conditional entropy variance, it was then sufficient to show the convergence of the same quantities for a single copy of a truncated output of the two-mode squeezed vacuum. Here, we use the same strategy, sending as input truncated two-mode squeezed vacuum states into the tensor product channel. By the additivity of relative entropy and conditional entropy variance, we can again simply invoke the convergence for each mode, which is proved in~\cite[Lemma 61]{mele2025achievableratesnonasymptoticbosonic}. Importantly, the crucial tool employed in the proof of \cite[Theorem~5]{mele2025achievableratesnonasymptoticbosonic} was \cite[Lemma~55]{mele2025achievableratesnonasymptoticbosonic}, which coincides with the above Lemma~\ref{thm_low_pure_loss_non_iid} in the particular case in which $\rho_i=\rho$ and $\sigma_i=\sigma$ for all $i=1,\ldots,n$.

    Hence, the proof of the lower bounds on the one-shot capacities in \eqref{eq_l_bound_one} is completely analogous to the one of \cite[Theorem~5]{mele2025achievableratesnonasymptoticbosonic}, with the only difference being that the former relies on Lemma~\ref{thm_low_pure_loss_non_iid} and the latter relies on \cite[Lemma~55]{mele2025achievableratesnonasymptoticbosonic}.

    Another slight difference regards the proof of the the lower bound on the one-shot quantum capacity in \eqref{eq_l_bound_one} in the particular case in which the condition $\lambda_1,\ldots,\lambda_n\ge \frac{1}{2}$ is not satisfied. Indeed, if $\lambda_1,\ldots,\lambda_n\ge \frac{1}{2}$, then for all $i=1,\ldots,n$ the coherent information of the generalised Choi state $\Id\otimes \mathcal{E}_{\lambda_i}(\Psi_{N_s})$ converges to the asymptotic quantum capacity $Q(\mathcal{E}_{\lambda_i})$ for $N_s\to\infty$, otherwise the former is negative and the latter vanishes. Hence, in the case in which the condition $\lambda_1,\ldots,\lambda_n\ge \frac{1}{2}$ is not satisfied, note that 
    \bb
Q^{(\varepsilon,1)}\left(\bigotimes_{i=1}^n\mathcal{E}_{\lambda_i}\right) &\ge Q^{(\varepsilon,1)}\left(\bigotimes_{\substack{i \in \{1,\ldots,n\}: \\ \lambda_i > \frac{1}{2}}} \mathcal{E}_{\lambda_i}\right)\\
&\ge \sum_{\substack{i \in \{1,\ldots,n\}: \\ \lambda_i > \frac{1}{2}}} Q(\mathcal{E}_{\lambda_i})-\log_2\frac{16}{16-\varepsilon}-\log_2 6- 2\log_2\frac{16+\varepsilon}{16-\varepsilon}-\log_2\!\left(\frac{2^{18}}{3\varepsilon^4}\right)\\
&= \sum_{i=1}^n Q(\mathcal{E}_{\lambda_i})-\log_2\frac{16}{16-\varepsilon}-\log_2 6- 2\log_2\frac{16+\varepsilon}{16-\varepsilon}-\log_2\!\left(\frac{2^{18}}{3\varepsilon^4}\right)\,,
\ee
where the second inequality follows by an entirely analogous reasoning to the one of \cite[Theorem~5]{mele2025achievableratesnonasymptoticbosonic}, and in the third equality we exploited that $Q(\mathcal{E}_{\lambda})=0$ for $\lambda\in[0,1/2]$.

\end{proof}

\subsection{Tensor product of arbitrary Gaussian channels}\label{app_sec_aep_gen}
In this subsection, we derive a simple lower bound on the one-shot capacities of tensor product of Gaussian channels by extending the asymptotic equipartition property approach developed in \cite{mele2025achievableratesnonasymptoticbosonic}. This lower bound turns out to be weaker than the one found in the previous subsection, and thus it is not used in the present paper. However, we remark that these results could be useful in estimating the \emph{energy-constrained} $n$-shot capacities of the model using the same approach explored in \cite{mele2025achievableratesnonasymptoticbosonic}. This subsection is reported for the sake of completeness and can be skipped by non-interested readers. For the definition of the entropic quantities used in this subsection, we refer the reader to Ref.~\cite[Section~A2]{mele2025achievableratesnonasymptoticbosonic}.

Let us start by proving a generalisation of the \emph{infinite-dimensional asymptotic equipartition property} (see~\cite[Theorem 6.2]{fawzi2023asymptotic}) to the case of product states. To prove this result, we will use the same approach as in \cite[Theorem 6.2]{fawzi2023asymptotic}, which generalises the result in \cite[Theorem 9]{Tomamichel_2009} to the infinite-dimensional setting.
\begin{lemma}[(Infinite-dimensional asymptotic equipartition property for product states)]\label{lemma_fawzi_noniid}
    Let $\rho\coloneqq\bigotimes_{i=1}^{n}\rho_i$ be a product state and let $\sigma\coloneqq\bigotimes_{i=1}^{n}\sigma_i$ be a a product of arbitrary positive semi-definite operators with bounded trace. For any $\varepsilon\in(0,1)$ with $n\ge2\log\left(\frac{2}{\varepsilon^2}\right)$ we have
    \bb
    D_{\max}^\varepsilon(\rho\|\sigma)\le \sum_{i=1}^{n}D(\rho_i\|\sigma_i)+4\sqrt{n}\log_2({\bar{\eta}_{\max}})\sqrt{\log\left(\frac{2}{\varepsilon^2}\right)}\,,
    \ee
    where
    \bb
        \bar{\eta}_i&\coloneqq \sqrt{2^{D_{3/2}(\rho_i\|\sigma_i)}}+\sqrt{2^{-D_{1/2}(\rho_i\|\sigma_i)}}+1\,,\\
        \bar{\eta}_{\max}&\coloneqq \max_{i\in\{1,\ldots,n\}}\bar{\eta}_i\,.
    \ee
Here, $D_{\max}^{\varepsilon}$, $D_{1/2}$, and $D_{3/2}$ denote the smooth max relative entropy, $\frac{1}{2}$-Petz Rényi relative entropy, and the $\frac{3}{2}$-Petz Rényi relative entropy~\cite{mele2025achievableratesnonasymptoticbosonic}.
\end{lemma}
\begin{proof}
Proceeding as in \cite[Theorem 6.2]{fawzi2023asymptotic}, for any $\alpha\in[1,2)$ we have that
\bb
D_{\max}^\varepsilon(\rho\|\sigma)&\le  D_{\alpha}(\rho\|\sigma)-\frac{1}{\alpha-1}\log_2(1-\sqrt{1-\varepsilon^2})\,,
\ee
where $D_\alpha$ denotes the $\alpha$-Petz Rényi relative entropy~\cite{mele2025achievableratesnonasymptoticbosonic}. Consequently, by exploiting the additivity of $D_\alpha$ and the fact that $\frac{\varepsilon^2}{2}\ge1-\sqrt{1-\varepsilon^2}$, we obtain that
\bb
D_{\max}^\varepsilon(\rho\|\sigma)&\le  \sum_{i=1}^n D_{\alpha}(\rho_i\|\sigma_i)+\frac{1}{\alpha-1}\log_2\!\left(\frac{2}{\varepsilon^2}\right)\,.
\ee
Moreover, as proved in \cite{fawzi2023asymptotic}, it holds that
\bb
    D_{\alpha}(\rho_i\|\sigma_i)\le D(\rho_i\|\sigma_i)+4(\alpha-1)(\log_2\eta_i)^2\,.
\ee
for any $\alpha$ such that $1\le\alpha< 1+\frac{\log_23}{4\log_2\eta_i}$. Hence, let us define
\bb
    \alpha\coloneqq 1+\frac{1}{2\lambda \sqrt{n}}
\ee
with $\lambda$ being such that $1\le\alpha< 1+\frac{\log_23}{4\log_2\eta_i}$ for all $i\in\{1,\ldots,n\}$. That is, $\lambda$ has to satisfy
\bb
    \frac{1}{\lambda\sqrt{n}}<\frac{\log 3}{2\log\bar{\eta}_i}
\ee
for all $i\in\{1,\ldots,n\}$. Hence, we have that
\bb
D_{\max}^\varepsilon(\rho\|\sigma)&\le  \sum_{i=1}^n \left[  D(\rho_i\|\sigma_i)+4(\alpha-1)(\log_2\eta_i)^2    \right]+\frac{1}{\alpha-1}\log_2\!\left(\frac{2}{\varepsilon^2}\right)\\
&\le  \sum_{i=1}^n D(\rho_i\|\sigma_i)+n4(\alpha-1)(\log_2\eta_{\max})^2    +\frac{1}{\alpha-1}\log_2\!\left(\frac{2}{\varepsilon^2}\right)\\
&=\sum_{i=1}^n D(\rho_i\|\sigma_i)+\frac{2\sqrt{n}}{\lambda}(\log_2\eta_{\max})^2    +2\lambda\sqrt{n}\log_2\!\left(\frac{2}{\varepsilon^2}\right)\,.
\ee
Now, we can optimise the above bound with respect $\lambda$ such that the constraint $\frac{1}{\lambda\sqrt{n}}<\frac{\log_2 3}{2\log_2\bar{\eta}_i}$ is satisfied for all $i\in\{1,\ldots,n\}$. To this end, we can choose
\bb
    \lambda\coloneqq\frac{\log_2\bar{\eta}_{\max}}{\sqrt{\log_2\frac{{2}}{\varepsilon^2}}}\,.
\ee
Indeed, this choice satisfies
\bb
    \frac{1}{\lambda\sqrt{n}}=\frac{\sqrt{{\log_2\frac{2}{\varepsilon^2}}}}{\log_2\bar{\eta}_{\max}\,\sqrt{n}}<\frac{\sqrt{5}}{2\sqrt{2}\log_2\bar{\eta}_{\max}}<\frac{\log_2 3}{2\log_2\bar{\eta}_i} 
\ee
for all $i\in\{1,\ldots,n\}$, where we used that $n\ge 2\log_2\frac{2}{\varepsilon^2}>\frac{8}{5}\log_2\frac{2}{\varepsilon^2}$. Therefore, by substituting this choice into the above bound, we conclude that

\bb
D_{\max}^\varepsilon(\rho\|\sigma)\le \sum_{i=1}^{n}D(\rho_i\|\sigma_i)+4\sqrt{n} \log_2 \bar{\eta}_{\max}\sqrt{\log_2\!\left(\frac{2}{\varepsilon^2}\right)}\,.
\ee
\end{proof}
We are now ready to present our lower bound on the one-shot quantum capacity of the tensor product of Gaussian channels.
\begin{thm}[(Lower bound on the one-shot quantum capacity of tensor product of Gaussian channels)]\label{thm_lower_bound_non_iid}
Let $\varepsilon\in(0,1)$, $\delta\in(0,1)$, $\eta\in[0,\frac{\varepsilon\sqrt{\delta}}{4})$, and $n\in\N$ such that $n\ge2\log_2\!\left(\frac{2}{\varepsilon^2}\right)$. Let 
\bb
    \NN\coloneqq \NN_{A'_1\to B_1}^{(1)}\otimes \NN_{A'_2\to B_2}^{(2)}\otimes \cdots\otimes \NN_{A'_n\to B_n}^{(n)}
\ee
be a tensor product of $n$ Gaussian channels and let $\Phi^{(1)}_{A_1 A'_1}$, $\Phi^{(2)}_{A_2 A'_2}$, $\ldots$, $\Phi^{(n)}_{A_n A'_n}$ be $n$ Gaussian states.
The one-shot quantum capacity of $\NN$ satisfies the following lower bound:
    \bb\label{lower_bound_CV3}
        Q^{(\varepsilon,1)}(\NN)\ge \sum_{i=1}^n I_c(A_i\,\rangle\, B_i)_{\Psi^{(i)}}-\sqrt{n}4\log_2({\bar{\eta}_{\max}})\sqrt{\log_2\!\left(\frac{2}{\left(\frac{\varepsilon\sqrt{\delta}}{4}-\eta \right)^2}\right)}+\log_2\!\left(\eta^4(1-\delta)\right)\,,
    \ee
    where
    \bb\label{eq_bar_eta3}
        \bar{\eta}_{\max}&\coloneqq \max_{i\in\{1,\ldots, n\}}\bar{\eta}_i\,,\\
        \bar{\eta}_i&\coloneqq \sqrt{2^{H_{1/2}(A_i|B_i)_{\Psi^{(i)}}}}+\sqrt{2^{H_{1/2}(A_i|E_i)_{\Psi^{(i)}}}}+1\,,\\
        \Psi^{(i)}_{A_iB_iE_i}&\coloneqq V^{(i)}_{A_i'\to B_iE_i}\Phi^{(i)}_{A_iA_i'}(V^{(i)}_{A_i'\to B_iE_i})^{\dagger}\,,
    \ee
    and $V^{(i)}_{A_i'\to B_iE_i}$ is a Stinespring dilation isometry associated with $\NN^{(i)}_{A_i'\to B_i}$. Moreover, $I_c$ and $H_{1/2}$ denote the coherent information and the conditional $\frac{1}{2}$-Petz Rényi entropies~\cite{mele2025achievableratesnonasymptoticbosonic}, respectively. In particular, by choosing $\delta=\frac{1}{4}$ and $\eta=\frac{\varepsilon}{16}$, we get
    \bb\label{lower_bound_CV4}
        Q^{(\varepsilon,1)}(\NN)\ge \sum_{i=1}^n I_c(A_i\,\rangle\, B_i)_{\Psi^{(i)}}-\sqrt{n}4\log_2({\bar{\eta}_{\max}})\sqrt{\log_2\!\left(\frac{2^9}{\varepsilon^2}\right)}-\log_2\!\left(\frac{2^{18}}{3\varepsilon^4}\right)\,.
    \ee
\end{thm}
\begin{proof}
    The proof is entirely analogous to the one of \cite[Theorem~45]{mele2025achievableratesnonasymptoticbosonic}. The key difference is that it relies on Lemma~\ref{lemma_fawzi_noniid} instead of \cite[Theorem 6.2]{fawzi2023asymptotic}.
\end{proof}
An analogous lower bound holds for the one-shot two-way quantum capacity and one-shot secret-key capacity of tensor product of Gaussian channels.
\begin{thm}[(Lower bound on the one-shot two-way quantum capacity and one-shot secret-key capacity of tensor product of Gaussian channels)]\label{thm_lower_bound_2way2}
Let $\varepsilon\in(0,1)$, $\eta\in[0,\sqrt{\varepsilon})$, and $n\in\N$ such that $n\ge2\log_2\!\left(\frac{2}{\varepsilon^2}\right)$.  Let 
\bb
    \NN\coloneqq \NN_{A'_1\to B_1}^{(1)}\otimes \NN_{A'_2\to B_2}^{(2)}\otimes \cdots\otimes \NN_{A'_n\to B_n}^{(n)}
\ee
be a tensor product of $n$ Gaussian channels and let $\Phi^{(1)}_{A_1 A'_1}$, $\Phi^{(2)}_{A_2 A'_2}$, $\ldots$, $\Phi^{(n)}_{A_n A'_n}$ be $n$ Gaussian states. Then, the one-shot two-way quantum capacity and the one-shot secret-key capacity of $\NN$ satisfy the following lower bounds:
    \bb\label{lower_bound_CV_2way5}
        K^{(\varepsilon,1)}(\NN)&\ge Q_2^{(\varepsilon,1)}(\NN)\ge \sum_{i=1}^n I_c(A_i\,\rangle\, B_i)_{\Psi^{(i)}}-\sqrt{n}4\log_2({\bar{\eta}_{1,\max}})\sqrt{\log_2\!\left(\frac{2}{\left(\sqrt{\varepsilon}-\eta \right)^2}\right)}+4\log_2\eta\,,\\
        K^{(\varepsilon,1)}(\NN)&\ge Q_2^{(\varepsilon,1)}(\NN)\ge \sum_{i=1}^n I_c(B_i\,\rangle\, A_i)_{\Psi^{(i)}}-\sqrt{n}4\log_2({\bar{\eta}_{2,\max}})\sqrt{\log_2\!\left(\frac{2}{\left(\sqrt{\varepsilon}-\eta \right)^2}\right)}+4\log_2\eta\,,
    \ee
    where
    \bb\label{eq_bar_eta_2way5}
        \bar{\eta}_{1,\max}&\coloneqq\max_{i\in\{1,\ldots,n\}}\bar{\eta}^{(i)}_{1}    \,,\\
        \bar{\eta}_{2,\max}&\coloneqq\max_{i\in\{1,\ldots,n\}}\bar{\eta}^{(i)}_{2}    \,,\\
        \bar{\eta}^{(i)}_1&\coloneqq \sqrt{2^{H_{1/2}(A_i|B_i)_{\Psi^{(i)}}}}+\sqrt{2^{H_{1/2}(A_i|E_i)_{\Psi^{(i)}}}}+1\,,\\
        \bar{\eta}^{(i)}_2&\coloneqq \sqrt{2^{H_{1/2}(B_i|A_i)_{\Psi^{(i)}}}}+\sqrt{2^{H_{1/2}(B_i|E_i)_{\Psi^{(i)}}}}+1\,,\\
        \Psi_{A_iB_iE_i}&\coloneqq V_{A_i'\to B_iE_i}\Phi_{A_iA_i'}(V_{A_i'\to B_iE_i})^{\dagger}\,,
    \ee
    and $V_{A'\to BE}$ is a Stinespring dilation isometry associated with $\NN_{A'\to B}$. Moreover, $I_c(A\,\rangle\, B)$, $I_c(B\,\rangle\, A)_{\Psi}$, $H_{1/2}(A|B)$ denote the coherent information, reverse coherent information, and conditional Petz Rényi entropies~\cite{mele2025achievableratesnonasymptoticbosonic}, respectively. In particular, by choosing $\eta=\frac{\sqrt{\varepsilon}}{2}$, we get
    \bb\label{b_lower_bound_CV_Q2_15}
        K^{(\varepsilon,1)}(\NN)&\ge Q_2^{(\varepsilon,1)}(\NN)\ge \sum_{i=1}^n I_c(A_i\,\rangle\, B_i)_{\Psi^{(i)}}-\sqrt{n}4\log_2({\bar{\eta}_{\max,1}})\sqrt{\log_2\!\left(\frac{8}{\varepsilon}\right)}-\log_2\!\left(\frac{16}{\varepsilon^2}\right)\,,\\
        K^{(\varepsilon,1)}(\NN)&\ge Q_2^{(\varepsilon,1)}(\NN) \ge \sum_{i=1}^n I_c(B_i\,\rangle\, A_i)_{\Psi^{(i)}}-\sqrt{n}4\log_2({\bar{\eta}_{\max,2}})\sqrt{\log_2\!\left(\frac{8}{\varepsilon}\right)}-\log_2\!\left(\frac{16}{\varepsilon^2}\right)\,.
    \ee
\end{thm}
\begin{proof}
    The proof is entirely analogous to the one of \cite[Theorem~48]{mele2025achievableratesnonasymptoticbosonic}. The key difference is that it relies on Lemma~\ref{lemma_fawzi_noniid} instead of \cite[Theorem 6.2]{fawzi2023asymptotic}.
\end{proof}

\newpage

\section{Ergodic Estimates for Toeplitz Sequences Generated by a Symbol}\label{ergodic_estimates}
This section is devoted to the proof of error bound on the Avram-Parter's theorem stated above in Theorem~\ref{th:final_estimation_lip_main}, which has played a crucial role in our analysis of non-asymptotic bosonic quantum communication with memory effects. To do this, we start by giving relevant preliminaries on matrix analysis and Toeplitz matrices.
\subsection{Preliminaries}
Throughout this section, a matrix-sequence is a sequence of the form $\{A_n\}_n$, where $A_n$ is a square matrix and ${\rm size}(A_n)=n$. Let $C_c(\mathbb R)$ 
be the space of continuous complex-valued functions with bounded support defined on $\mathbb R$. 
If $A\in\mathbb C^{n\times n}$, the singular values and eigenvalues of $A$ are denoted by $\sigma_1(A),\ldots,\sigma_n(A)$ and $\lambda_1(A),\dots,\lambda_n(A)$, respectively. 
We denote by $\mu_d$ the Lebesgue measure in $\mathbb R^d$. Throughout this section, ``measurable'' means ``Lebesgue measurable''. Moreover, we denote by $\|g\|_p$ the $L^p$-norm of the function $g$ over its domain (which will be clear from the context). 

We say that a matrix-sequence $\{A_n\}_n$ admits a singular value symbol (or simply a symbol) $f(x)$ when the sampling of $|f(x)|$  over a   uniform   grid in   its domain yields an approximation of the singular values of   $A_n$ that improves   as $n\to \infty$. The  rigorous definition is given below; it is based   on an ergodic formula that must hold for every test   function $F\in C_{c}(\mathbb R)$.   

\begin{Def}[(Asymptotic singular value distribution of a matrix-sequence)]\label{dd}
Let $\{A_n\}_n$ be a matrix-sequence with $A_n$ of size $n$, and let $f:\Omega\subset\mathbb R^d\to\mathbb C$ be measurable with $0<\mu_d(\Omega)<\infty$.
We say that $\{A_n\}_n$ has an asymptotic singular value distribution described by $f$ if the following \textit{ergodic formula} holds:
\begin{equation}\label{eq:ergodic_formula}
\lim_{n\to\infty}\frac1{n}\sum_{i=1}^{d_n}F(\sigma_i(A_n))=\frac1{\mu_d(\Omega)}\int_\Omega F(|f(x)|)\mathrm{d}x,\qquad\forall\,F\in C_c(\mathbb R).
\end{equation}
In this case, $f$ is called the singular value symbol (or simply the symbol) of $\{A_n\}_n$ and we write $\{A_n\}_n\sim_\sigma f$.
\end{Def}

For a comprehensive treatment of the topic of symbols, 
a tutorial on how to compute them, a showcase of their applications to several problems (mostly coming from discretization of PDEs and iterative numerical methods), the main references are \cite{GLT1,GLT2,GLT3,GLT4,redGLT,recGLT}. 

Among the various matrix-sequences that possess a symbol, Toeplitz sequences generated by $L^1$ functions stand out for their importance.  
A matrix of the form
\begin{equation}\label{mbtm_expr}
\left[f_{i-j}\right]_{i,j=1}^{n}=\begin{bmatrix}
f_0 & f_{-1} & \ f_{-2} & \ \cdots & \ \ \cdots & f_{-(n-1)} \\
f_1 & \ddots & \ \ddots & \ \ddots & \ \ & \vdots\\
f_2 & \ddots & \ \ddots & \ \ddots & \ \ \ddots & \vdots\\
\vdots & \ddots & \ \ddots & \ \ddots & \ \ \ddots & f_{-2}\\
\vdots & & \ \ddots & \ \ddots & \ \ \ddots & f_{-1}\\
f_{n-1} & \cdots & \ \cdots & \ f_2 & \ \ f_1 & f_0
\end{bmatrix},
\end{equation}
whose $(i,j)$-th entry depends only on the difference $i-j$ is called a Toeplitz matrix. In other words, a Toeplitz matrix is a matrix whose entries are constant along each diagonal.
A case of special interest arises when the entries $f_k$ are the Fourier coefficients of a function $f:[-\pi,\pi]\to\mathbb C$ in $L^1([-\pi,\pi])$, i.e.,
\[ f_k=\frac1{2\pi}\int_{-\pi}^\pi f(x){\rm e}^{-{\rm i}kx}{\rm d}x,\qquad k\in\mathbb Z. \]
In this case, the matrix \eqref{mbtm_expr} is denoted by $T_n(f)$ and is referred to as the $n$-th Toeplitz matrix generated by $f$. One can also define the associated infinite-dimensional operator $T(f)$ and represent it as the infinite Toeplitz matrix $\left[f_{i-j}\right]_{i,j\in \mathbb N}$.   
It is not difficult to see that the conjugate transpose of $T_n(f)$ is given by
\begin{equation*}
T_n(f)^H=T_n(\overline f)
\end{equation*}
for every $f\in L^1([-\pi,\pi])$ and every $n$; see, e.g., \cite[Section~6.2]{GLT1}. In particular, $T_n(f)$ is Hermitian whenever $f$ is almost everywhere a real-valued function.\\

The asymptotic singular value distribution of the Toeplitz sequence $\{T_n(f)\}_n$ has been deeply investigated over time, starting from Szeg\"o’s first limit theorem \cite{GS} and the Avram–Parter theorem \cite{Avram,Parter}, up to the works by Tyrtyshnikov–Zamarashkin \cite{Ty96,ZT,ZT'} and Tilli \cite{TilliL1,Tilli-complex}. For more on this subject, see \cite[ Chapter 5]{BoSi}  and \cite[Chapter 6]{GLT1} and references therein. 
Using our notations, the Avram-Parter's theorem states that $\{T_n(f)\}_n\sim_\sigma f$ for any  $f\in L^1([-\pi,\pi])$. As a consequence, the ergodic formula \eqref{eq:ergodic_formula} holds with $A_n=T_n(f)$ for any $L^1$ function $f$.

The analysis of the ergodic formula for Toeplitz sequences generated by real-valued symbols $f$ is also linked to the asymptotic study of the eigenvalues and spectral distribution of the matrices $T_n(f)$. Several works enquire the existence of an asymptotic expansion of the eigenvalues when we impose certain conditions on $f$. 
See \cite{SL} and references therein for the case of (modified) simple loop functions, that are in particular nonnegative $C^1_{per}[-\pi,\pi]$ functions with $f(0)=f'(0)=0$ being a zero of order at most 2 and exactly another extremum $f'(\xi)=0$ in the open interval $(-\pi,\pi)$. 
For the case of (some) smooth even function $f$ that is monotonous on $[0,\pi]$ see \cite{Sven1,Sven2}. 
About the classes of test functions satisfying the ergodic formula \eqref{eq:ergodic_formula}  see for example \cite{Serra02,Serra03}. 

In this context, Widom proved a more precise result \cite{Widom}. He showed that, for essentially bounded symbols $f$ with $$\vertiii{f}^2\coloneqq \sum_{k=-\infty}^{\infty} \lvert k\rvert \lvert f_k\rvert^2<\infty,$$ we have 
\begin{equation}\label{eq:ergodic_Widom}
    \lim_{n\to \infty} \left\{ 
    \sum_{j=1}^n G(\sigma_j(T_n(f))^2) - \frac n {2\pi} \int_{-\pi}^\pi G(\lvert f(\theta)\rvert^2)d\theta
    \right\} = \Tr [G(T(\overline f)T(f)) + G(T(f)T(\overline f)) - 2T(G(\lvert f\rvert^2)) ]
\end{equation}
for any $C^3$ function $G(x)$ defined at least on an interval containing  $\sigma_j(T_n(f))^2$ and $\sigma_j(T(f))^2$  for all $j,n$.  This implies that 
\[\left\lvert 
\frac 1n
\sum_{j=1}^n F(\sigma_j(T_n(f))) - \frac 1 {2\pi} \int_{-\pi}^\pi F(\lvert f(\theta)\rvert)d\theta 
\right\rvert = O\left(\frac 1n\right)
\]
thus providing an estimate for the convergence rate of the ergodic formula \eqref{eq:ergodic_formula} for any function $F(x)\in C_c(\mathbb R)$ that is an extension a function of the form $G(x^2)$ with $G(x)$ as specified above. 

For concrete applications, we need an explicit bound of the form
\begin{equation} \label{eq:ergodic_bound}
\left\lvert 
\frac 1n
\sum_{j=1}^n F(\sigma_j(T_n(f))) - \frac 1 {2\pi} \int_{-\pi}^\pi F(\lvert f(\theta)\rvert)d\theta 
\right\rvert \le c(n)
\end{equation}
holding for any $n\ge n_0$ with an explicit $n_0$. Equation \eqref{eq:ergodic_Widom} shows that, for some choices of $F(x)$ and $f(x)$,  we can take $c(n) = O(1/n)$, but the dependence of $c(n)$ on $f$ and $F$ is not specified. Moreover, this result holds only for regular enough test functions $F(x)$.

In the main result of this section (Theorem~\ref{th:final_estimation_lip} below), we provide an explicit bound $c(n)$ in \eqref{eq:ergodic_bound} for the case of test functions $F(x)$ in $C_c(\mathbb R)$ that are not necessarily continuously differentiable, but at the cost of requiring that the symbol $f(x)$ belongs to $C^k_{per}[-\pi,\pi]$ for some $k\ge 1$, i.e., $f(-\pi)=f(\pi)$ and the $2\pi$-periodic extension of $f(x)$ to $\mathbb R$ is a function in $C^k(\mathbb R)$ for some $k\ge 1$. In this scenario, the bound $c(n)$ is $O(1/n^{\frac  {2k+1}{2k+3}})$ and it is derived by approximating $f(x)$ with its truncated Fourier series $f_N(x)$ and the singular values of the banded Toeplitz matrices $T_n(f_N)$ with the ones of the circulant matrices $C_n(f_N)$,   which are exactly the sampling of $|f(x)|$ over a uniform grid in $[-\pi,\pi]$. Here we denote the set of real-valued Lipschitz functions with Lipschitz constant $L$ as $Lip(L)$, i.e. all functions $f:\mathbb R\to\mathbb C$ such that
\[
|f(x)-f(y)|\le L|x-y| \quad \forall x,y\in \mathbb R.
\]

Let us state our main result.
\begin{thm}\label{th:final_estimation_lip}
Let $f\in C^k_{per}[-\pi,\pi]$ with $k\ge 1$  and let $F\in C_c(\mathbb R)\cap Lip(L)$. Then
$$\left\lvert \frac 1n\sum_{i=1}^{n} F(\sigma_i(T_n(f))) - 
\frac 1{2\pi}\int_0^{2\pi} F(\lvert f(\theta) \rvert)d\theta
\right\rvert\le 
 \left(
\frac {2^{k+1}\|f^{(k)}\|_2L}{\sqrt{2\pi}}  
+    
 4\pi L \|f\|_2
\right)
\frac1{n^{\frac {k}{k+3/2}}}
+ 
(2\|F'\|_1 
 + 
4\|F\|_\infty)
\frac{1}{n^{\frac {k+1/2}{k+3/2}}}
$$
for every $n\ge 4$.  
\end{thm}

\subsection{Known results}

\subsubsection{Matrix norms and Toeplitz matrices}
Given a matrix $A\in \mathbb C^{n\times n}$,  we denote its $p$-Schatten norms as $\|A\|_p$ with $\infty \ge p\ge 1$, where $\|A\|_1$ is the Trace norm, $\|A\|_2$ is the Hilbert-Schmidt norm and $\|A\|_\infty = \|A\|$ is the operator norm. 
An important result that links together the different norms is given by the H\"older theorem.
\begin{thm}[(H\"older)]
\label{th:H\"older}If $\frac 1r = \frac 1p+\frac 1q$ with $\infty\ge r,p,q\ge 1$, then
\[
\|AB\|_r\le \|A\|_p\|B\|_q
\]
where by convention, $1/\infty =0$. 
In particular, for any $\infty \ge p\ge 1$
\begin{equation*}\label{eq:trace_HS_norm}
    \|AB\|_p \le \|A\|\|B\|_p, \qquad 
\lvert \Tr(AB)\rvert \le \|AB\|_1\le \|A\|_2\|B\|_2.
\end{equation*}
\end{thm}

Here we report a classic result linking Toeplitz matrices to their symbols through their $p$-Schatten norm and $p$-norm. 
\begin{thm}[(Theorem 6.2 of \cite{GLT1})]\label{th:norm_toeplitz_symbol} 
    Let $f\in L^p([-\pi,\pi])$ for some $p\in [1,\infty]$. Then
    $$\|T_n(f)\|_p\le \left(\frac n{2\pi}\right)^{1/p} \|f\|_p$$
    where for any matrix $A$,  $\|A\|_p$ is its $p$-Schatten norm. 
\end{thm}





\subsubsection{Lipschitz and BV functions}

\begin{Def}
A function $f$ is called $Lip(M,\alpha)$ with $M\ge 0$, $\alpha\in (0,1]$, if 
\[
\lvert f(x) - f(y)\rvert \le M \lvert x-y\rvert^\alpha
\]
for every $x,y$ in the domain of $f$.  Notice that $Lip(M)=Lip(M,1)$.
\end{Def}

If $D$ is a bounded domain in $\mathbb R$ and $f\in C^1(D)$, then $f\in Lip(\|f'\|_\infty)$. 
On the other hand, any $f\in Lip(L)$ on $D$ is also in $BV(D)$ with total variation bounded by $L\lvert D\rvert$.
Moreover, the triangular inequality
$$ \lvert \lvert f(x)\rvert - \lvert f(y)\rvert\rvert \le \lvert f(x)-f(y) \rvert$$
is enough to conclude that $f\in Lip(M,\alpha) \implies \lvert f\rvert \in Lip(M,\alpha)$. 

\begin{thm}{(Theorem 4.3 of \cite{Serra01})}\label{th:Lip_test_function}
Let $\serie A$ and $\serie B$ be two sequences of $d_n\times d_n$ matrices and assume that $F\in Lip(M,\alpha)$ is a continuous function. Then
$$\left\lvert \frac 1n\sum_{i=1}^{d_n} F(\sigma_i(A_n)) - \frac 1n\sum_{i=1}^{d_n} F(\sigma_i(B_n)) \right\rvert \le M \|A_n-B_n\|_{p}^\alpha d_n^{-\alpha/p}$$
where $\|\cdot\|_p$ is the $p$-Schatten norm for any $p\in [1,\infty]$.
\end{thm}

Lipschitz functions are in particular \textit{Bounded Variation}, so it can be decomposed into the difference of two continuous monotonic increasing functions as showed by the following classical result. 

\begin{lemma}\label{lem:monotone_C1}
Given $F\in BV([a,b])$ there exist $H,G\in C( [a,b])$ monotonous increasing functions such that 
$$F = G - H,\qquad \|G\|_\infty + \|H\|_\infty = TV(F) + F(a).$$
In particular, $G(a)=F(a)$ and $H(a) = 0$. 
If in addition $F\in Lip(L)$, then $F'$ exists almost everywhere, $TV(F) = \|F'\|_1\le L(b-a)$ and $H, G$ are continuous.
\end{lemma}

\subsubsection{Estimates for the Fourier transform and series}

For the main results, we need some estimates on the moments of the Fourier transform of $C^3$ functions, and on the truncated Fourier series of $C^k_{per}[-\pi,\pi]$ functions. 


\begin{lemma}\label{lem:Fourier_approx}
Suppose $f:\mathbb R \to \mathbb C$ is $2\pi$-periodic and in $C^k(\mathbb R)$. Call $f_N$ its truncated Fourier series at degree (both negative and positive) $N\ge 0$. Then
$$\|f-f_N\|_2^2 \le \frac {\|f^{(k)}\|_2^2}{N^{2k}}$$
\end{lemma}
\begin{proof}
If $c_n$ are the Fourier coefficients of $f$, then the Fourier coefficients of $f^{(k)}$ are $(\textnormal in)^kc_n$, so
$$
\|f-f_N\|_2^2 = 2\pi \sum_{\lvert n\rvert>N} \lvert c_n\rvert^2 = 2\pi \sum_{\lvert n\rvert>N} \frac{|(\textnormal i n)^k c_n\rvert ^2}{n^{2k}}
\le \frac {\|f^{(k)}\|_2^2}{N^{2k}}.
$$
\end{proof}


\begin{lemma}
    \label{lem:inf_norm_vs_L2_norm_antiderivative}
   Suppose $f:\mathbb R \to \mathbb C$ is $2\pi$-periodic and in $C^1(\mathbb R)$. Then for any $N$,
 \[
  \|(f_N)'\|_\infty \le \sqrt{ 4\pi}N^{3/2} \|f\|_2\le 4N^{2} \|f\|_2.
   \]
\end{lemma}
\begin{proof}
   If $c_n$ are the Fourier coefficients of $f$, then $\textnormal i nc_n$ are the Fourier coefficients of $f'$ and by Cauchy-Schwartz,
      \[
   \|(f_N)'\|_\infty^2 \le \left( \sum_{|n|\le N} |nc_n|\right)^2 \le  \sum_{|n|\le N} |c_n|^2 \sum_{|n|\le N} n^2 \le 4\pi\|f\|_2^2 \frac{N(N+1)(2N+1)}{6} .
   \]
Since
   \[
   N\ge 1\implies 0\le (N-1)(4N+1) = 6N^2 - (N+1)(2N+1)
\implies 
\frac{N(N+1)(2N+1)}3 \le 2N^3
   \]
then the thesis holds for $N\ge 1$. To conclude, notice that for $N=0$ the thesis reads $0\le 0$. 
\end{proof}

\subsection{Lipschitz Test Functions}
This section is dedicated to the proof of Theorem \ref{th:final_estimation_lip}. The idea of the proof is to uniformly approximate the periodic symbol $f$ with its truncated Fourier series $f_N$ and consequently to approximate the Toeplitz matrices $T_n(f)$ and $T_n(f_N)$ with the related circulant matrix $C_n(f_N)$, defined as

\bb
    (C_n(f_N))_{i-j} \coloneqq
    \begin{cases}
    (T_n(f))_{i-j}\,,  & |i- j|\leq 2N,\\
    (T_n(f))_{-n+(i-j)}\,, & |i- j|> 2N \cap i>j,  \\
    (T_n(f))_{n+(i-j)} , & |i- j|> 2N \cap i<j.
    \end{cases}
\ee

 The circulant matrix has singular values $\sigma_i\coloneqq \sigma_i(C_n(f_N)) = |f_N(2\pi i/n)|$, so the average of $F(\sigma_i)$ is a quadrature formula for the function $F(|f_N|)/(2\pi)$. The proof can be thus divided into four steps:
\begin{enumerate}
    \item using that $f_N$ is an approximation of $f$, find a bound $c_1(n,N)$ for
\[\left\lvert \frac 1n\sum_{i=1}^{n} F(\sigma_i(T_n(f))) - 
\frac 1n\sum_{i=1}^{n} F(\sigma_i(T_n(f_N)))
\right\rvert\le c_1(n,N),\]
\item using that the rank of $C_n(f_N)-T_n(f_N)$ is $O(N)$ for $N<n/2$, find a bound $c_2(n,N)$ for
\[\left\lvert \frac 1n\sum_{i=1}^{n} F(\sigma_i(T_n(f_N))) - 
\frac 1n\sum_{i=1}^{n} F(\sigma_i(C_n(f_N)))
\right\rvert\le c_2(n,N),\]
\item determine a quadrature bound $c_3(n,N)$ for
\[\left\lvert 
\frac 1n\sum_{i=1}^{n} F(\sigma_i(C_n(f_N)))
- 
\frac 1{2\pi}\int_0^{2\pi} F(\lvert f_N(\theta) \rvert)d\theta
\right\rvert\le c_3(n,N),\]
\item using that $f_N$ is an approximation of $f$, find a bound $c_4(N)$ for
\[\left\lvert 
\frac 1{2\pi}\int_0^{2\pi} F(\lvert f_N(\theta) \rvert)d\theta
-
\frac 1{2\pi}\int_0^{2\pi} F(\lvert f(\theta) \rvert)d\theta
\right\rvert\le c_4(N).\]
\end{enumerate}
The sum of the four quantities $c(n,N)\coloneqq c_1(n,N)+c_2(n,N)+c_3(n,N)+c_4(N)$ is thus a bound for
\begin{equation}
    \label{eq:cn}
    \left\lvert \frac 1n\sum_{i=1}^{n} F(\sigma_i(T_n(f))) - 
\frac 1{2\pi}\int_0^{2\pi} F(\lvert f(\theta) \rvert)d\theta
\right\rvert\le c(n,N) 
\end{equation}
and it holds for any $N<n/2$, so the wanted bound will be $\min_{N<n/2} c(n,N)$.  

\subsubsection{Step 1}

\begin{lemma}\label{lem:from_toep_to_toep_banded}
Let $f\in C^k_{per}[-\pi,\pi]$, $F\in Lip(M,\alpha)$ and $N\ge 0$. 
Then
$$ \left\lvert \frac 1n\sum_{i=1}^{n} F(\sigma_i(T_n(f))) - \frac 1n\sum_{i=1}^{n} F(\sigma_i(T_n(f_N))) \right\rvert\le M 
\left(\frac 1{2\pi} \frac {\|f^{(k)}\|^2_2}{N^{2k}} \right)^{\alpha/2}.$$
\end{lemma}
\begin{proof}
Call $\sigma_i$ the singular values of $T_n(f)$ and $\sigma_i^{(N)}$ the singular values of $T_n(f_N)$. From Theorem \ref{th:Lip_test_function}, Theorem \ref{th:norm_toeplitz_symbol}, and Lemma \ref{lem:Fourier_approx},
\begin{align*}
    \left\lvert \frac 1n\sum_{i=1}^{n} F(\sigma_i) - \frac 1n\sum_{i=1}^{n} F(\sigma_i^{(N)}) \right\rvert &\le M \|T_n(f-f_N)\|_{2}^\alpha n^{-\alpha/2}
\\ &\le 
M 
\left(\frac 1{2\pi}\right)^{\alpha/2} \|f-f_N\|_2^\alpha \\
&\le 
M 
\left(\frac 1{2\pi} \frac {\|f^{(k)}\|^2_2}{N^{2k}} \right)^{\alpha/2}
\end{align*}
\end{proof}
Since in the hypothesis of Theorem \ref{th:final_estimation_lip}, $F\in Lip(L)=Lip(L,1)$, we conclude that for any $N$
\begin{equation}
    \label{eq:c1}
    c_1(n,N)= L 
\left(\frac 1{2\pi}\right)^{1/2} \frac {\|f^{(k)}\|_2}{N^{k}}. 
\end{equation}

\subsubsection{Step 2}

\begin{lemma}\label{lem:low_rank_perturbation_Lip}
Suppose $A,B$ are $n\times n$ matrices and $F\in Lip(M)\cap C_c(\mathbb R)$. 
Then
$$\frac 1n \left\lvert \sum_{i=1}^{n} F(\sigma_i(A)) -  F(\sigma_i(B)) \right\rvert\le \|F'\|_1\frac{\rk(A-B)}{n}.$$
\end{lemma}
\begin{proof}
    Let $D\coloneqq \rk(A-B)$. By the Interlacing Theorem for singular values \cite[Theorem 2.11]{GLT1}, we know that
    \begin{equation}\label{eq:interlacing_sv}
        \sigma_{i-D}(A)\ge \sigma_i(B)\ge \sigma_{i+D}(A)
    \end{equation}
 for any $n-D\ge i>D$, where the singular values are all sorted in decreasing order.  Let $[a,b]$ be an interval containing the support of $F$. Thanks to Lemma \ref{lem:monotone_C1}, we know that $F=G-H$ on $[a,b]$ with $G,H$ continuous, nonnegative (since $G(a)=F(a)=H(a)=0$), monotonously increasing and with $\|G\|_\infty + \|H\|_\infty =\|F'\|_1$. Extend $G,H$ so that they are constant outside $[a,b]$, while still being continuous. 
$$ \frac 1n\left\lvert \sum_{i=1}^{n} F(\sigma_i(A)) - F(\sigma_i(B)) \right\rvert\le \frac 1n
 \left\lvert \sum_{i=1}^{n} G(\sigma_i(A)) -  G(\sigma_i(B)) \right\rvert
 +\frac 1n
 \left\lvert \sum_{i=1}^{n} H(\sigma_i(A)) -  H(\sigma_i(B)) \right\rvert.$$
Let us analyse the term with $G$,  the other one is analogous. Using \eqref{eq:interlacing_sv},
\begin{align*}
\sum_{i=1}^{n} G(\sigma_i(A)) - G(\sigma_i(B))
& = \sum_{i=1}^{D} G(\sigma_{i}(A)) - 
\sum_{i=n-D+1}^{n} G(\sigma_{i}(B))
+ \sum_{i=1}^{n-D} \Big[ G(\sigma_{i+D}(A)) - G(\sigma_i(B))\Big]&&\le 
D\|G\|_\infty,\\
\sum_{i=1}^{n} G(\sigma_i(A)) - G(\sigma_i(B))
& = \sum_{i=n-D+1}^{n} G(\sigma_{i}(A)) - 
\sum_{i=1}^{D} G(\sigma_{i}(B))
+ \sum_{i=D+1}^{n}\Big[ G(\sigma_{i-D}(A)) - G(\sigma_{i}(B))\Big]&&\ge
- D \|G\|_\infty,
\end{align*}
thus
$$\frac 1n
 \left\lvert \sum_{i=1}^{n} G(\sigma_i(A)) -  G(\sigma_i(B)) \right\rvert
\le \frac{D\|G\|_\infty}{n}$$
and
$$\frac 1n\left\lvert \sum_{i=1}^{n} F(\sigma_i(A)) - F(\sigma_i(B)) \right\rvert\le D\frac{\|G\|_\infty+\|H\|_\infty}{n}=  \frac{D\|F'\|_1}{n}.$$
\end{proof}

Since  $\rk(T_n(f_N)-C_n(f_N))\le 2N$ for $N<n/2$, we can apply the above result and obtain
\begin{equation}
    \label{eq:c2}
    c_2(n,N)= 2N\frac{\|F'\|_1}{n}. 
\end{equation}


\subsubsection{Step 3}

\begin{lemma}\label{lem:rectangle_rule_error_Lip}
Let $F\in Lip(L) \cap C_c(\mathbb R)$ and let $g$ be a complex trigonometric polynomial of degree $d$. Then for $n>2d$
$$ \left\lvert \frac 1n \sum_{i=1}^{n} F(\sigma_i(C_n(g))) - 
    \frac 1{2\pi}\int_0^{2\pi} F(\lvert g(\theta) \rvert)d\theta
    \right\rvert\le  \frac {\pi L}{n}\|g'\|_\infty + 4d\frac{\|F\|_\infty}{n}
    .$$
\end{lemma}
\begin{proof}
If $n>2d$, notice that the singular values of $C_n(g)$ are, up to the order, $g(2\pi i/n)$, so we have to estimate
$$\left\lvert \frac 1n \sum_{i=1}^{n} F\left(\left\lvert
g\left(\frac{2\pi i}{n} \right)
\right\rvert\right) - 
    \frac 1{2\pi}\int_0^{2\pi} F(\lvert g(\theta) \rvert)d\theta
    \right\rvert$$
that can be seen as the error of the rectangle rule applied to the function $F(\lvert g(\theta)\rvert)$.

The function $g(\theta)$ can be expressed as a Laurent polynomial in $z=e^{\textnormal i\theta}$ of degree $d$, thus it can have at most $2d$ zeros. Let $I_i \coloneqq [i-1, i]\cdot 2\pi/n$, and notice that if $g(\theta)$ is never zero in $I_i$, then $\lvert g(\theta)\rvert \in C^1(I_i)$ so we can expand by Taylor $\lvert g((i-1)2\pi/n + \theta)\rvert  =\lvert g((i-1)2\pi/n)\rvert  \pm  \theta   g'(\xi_\theta)$ where $\xi_\theta\in I_i$ depends on $\theta$. Using $F\in Lip(L)$, we get
\begin{align*}
    \left\lvert \frac 1n  F\left(\left\lvert
g\left(\frac{2\pi (i-1)}{n} \right)
\right\rvert\right) - 
    \frac 1{2\pi}\int_{\frac{2\pi (i-1)}{n}}^{\frac{2\pi i}{n}} F(\lvert g(\theta) \rvert)
    d\theta
    \right\rvert &=
       \left\lvert 
    \frac 1{2\pi}\int_{\frac{2\pi (i-1)}{n}}^{\frac{2\pi i}{n}}
      F\left(\left\lvert
g\left(\frac{2\pi (i-1)}{n} \right)
\right\rvert\right)
-
 F(\lvert g(\theta) \rvert)
    d\theta
    \right\rvert\\
    &\le 
    \frac 1{2\pi}
    \int_{\frac{2\pi (i-1)}{n}}^{\frac{2\pi i}{n}} 
  L \left\lvert 
    g(\theta) 
-
g\left(\frac{2\pi (i-1)}{n} \right)
   \right\rvert d\theta
    \\
    &= 
\frac L{2\pi}\int_{0}^{\frac{2\pi }{n}}  \theta \lvert g'(\xi_\theta) \rvert d\theta
     \le 
 \frac {\pi L}{n^2}\|g'\|_\infty. 
\end{align*}
If instead $g(\theta)$ has a zero in $I_i$ one can use the more immediate bound
$$ \left\lvert \frac 1n  F\left(\left\lvert
g\left(\frac{2\pi (i-1)}{n} \right)
\right\rvert\right) - 
    \frac 1{2\pi}\int_{\frac{2\pi (i-1)}{n}}^{\frac{2\pi i}{n}} F(\lvert g(\theta) \rvert)d\theta
    \right\rvert 
    \le 
2 \frac{\|F\|_\infty}{n}   $$
thus concluding that 
$$\left\lvert \frac 1n \sum_{i=1}^{n} F\left(\left\lvert
g\left(\frac{2\pi i}{n} \right)
\right\rvert\right) - 
    \frac 1{2\pi}\int_0^{2\pi} F(\lvert g(\theta) \rvert)d\theta
    \right\rvert
\le 
 \frac {\pi L}{n}\|g'\|_\infty + 4d\frac{\|F\|_\infty}{n}.
$$

\end{proof}

Using Lemma \ref{lem:inf_norm_vs_L2_norm_antiderivative} on $g=f_N$, the result shows that when $N<n/2$ and $f\in C^1_{per}[-\pi,\pi]$,
\begin{equation}
    \label{eq:c3}
    c_3(n,N)=  4N\sqrt N \frac {\pi L}{n}\|f\|_2 + 4N\frac{\|F\|_\infty}{n}. 
\end{equation}

\subsubsection{Step 4}

\begin{lemma}\label{lem:from_symbol_to_truncated_symbol}
Let $f\in C^k_{per}[-\pi,\pi]$, $F\in Lip(L)$ and $N\ge 0$. 
Then
$$ \left\lvert \frac 1{2\pi}\int_0^{2\pi} F(\lvert f(\theta) \rvert)d\theta
-
\frac 1{2\pi}\int_0^{2\pi} F(\lvert f_N(\theta) \rvert)d\theta\right\rvert\le 
\frac L{\sqrt{2\pi}}\frac {\|f^{(k)}\|_2}{N^{k}} 
.$$
\end{lemma}
\begin{proof}
By triangular inequality and Cauchy-Schwarz inequality, 
\begin{align*}
    \left\lvert \frac 1{2\pi}\int_0^{2\pi} F(\lvert f(\theta) \rvert)d\theta
-
\frac 1{2\pi}\int_0^{2\pi} F(\lvert f_N(\theta) \rvert)d\theta\right\rvert
&\le \frac 1{2\pi}\int_0^{2\pi}
 \left\lvert  F(\lvert f(\theta) \rvert) - F(\lvert f_N(\theta) \rvert)\right\rvert d\theta\\
 &\le \frac L{2\pi}\int_0^{2\pi}
 \left\lvert  \lvert f(\theta) \rvert - \lvert f_N(\theta) \rvert \right\rvert d\theta\\
  &\le \frac L{2\pi}\int_0^{2\pi}
 \left\lvert  f(\theta) -  f_N(\theta) \right\rvert d\theta
 \\
 &= \frac L{2\pi} \|f-f_N\|_1 \le \frac L{\sqrt{2\pi}} \|f-f_N\|_2.
\end{align*}
Eventually, Lemma \ref{lem:Fourier_approx} shows that 
\[
\|f-f_N\|_2 \le \frac {\|f^{(k)}\|_2}{N^{k}} 
\]
concluding the proof.
\end{proof}
The result shows that for any $N$, 
\begin{equation}
    \label{eq:c4}
    c_4(N)= \frac L{\sqrt{2\pi}}\frac {\|f^{(k)}\|_2}{N^{k}} . 
\end{equation}

\subsubsection{Finalization of the proof}

Take an integer $N<\frac n2$. From the relations \eqref{eq:cn}, \eqref{eq:c1}, \eqref{eq:c2}, \eqref{eq:c3} and \eqref{eq:c4},  we find the bound 
\[
\left\lvert \frac 1n\sum_{i=1}^{n} F(\sigma_i(T_n(f))) - 
\frac 1{2\pi}\int_0^{2\pi} F(\lvert f(\theta) \rvert)d\theta
\right\rvert \le c(n,N)
\]
\begin{align}\label{eq:final:bound_cnN}
c(n,N) =
 \frac {L}{\sqrt{2\pi}} \frac {\|f^{(k)}\|_2}{N^{k}} +    
\|F'\|_1 \frac{2N}{n}
 + 
 4N\sqrt N\frac {\pi L}{n}\|f\|_2 + 4N\frac{\|F\|_\infty}{n}
+ 
\frac {L}{\sqrt{2\pi}} \frac {\|f^{(k)}\|_2}{N^{k}}.
\end{align}
Since the previous formula holds for any $N< \frac{n}{2}$, we can minimise it over $N$. The main non-constant terms are $N\sqrt N/n$ and $1/N^k$, so we look for some $\alpha\in (0,1)$ and $N\sim n^\alpha$ such that the two terms are balanced, i.e. $n^{3\alpha/2-1} = n^{-\alpha k}$. In fact, the choice $\alpha = 1/(k+3/2)$ leads to 
$$\left\lvert \frac 1n\sum_{i=1}^{n} F(\sigma_i(T_n(f))) - 
\frac 1{2\pi}\int_0^{2\pi} F(\lvert f(\theta) \rvert)d\theta
\right\rvert = O\left( n^{- \frac{k}{k+3/2}} \right).$$
More accurately, set $N=\lfloor n^{\frac 1{k+3/2}}\rfloor$, and notice that since $k\ge 1$,
\[
n\ge 4\implies n^5 > 32n^2 \implies \frac n2 > n^{\frac 25} \ge  n^{\frac 1{k+3/2}} \ge \lfloor n^{\frac 1{k+3/2}}\rfloor=N.
\]
and that $n\ne 0\implies N\ge 1\implies N\ge (N+1)/2$, so
$$\frac 12 n^{\frac 1{k+3/2}}\le \frac 12(\lfloor n^{\frac 1{k+3/2}}\rfloor+1)\le N \le n^{\frac 1{k+3/2}}.$$
As a consequence, for any $n>3$ and $k\ge 1$, 
\begin{align}\label{eq:final_estimation}
\min_{N<n/2}c(n,N) \le 
 \left(
\frac {2^{k+1}\|f^{(k)}\|_2L}{\sqrt{2\pi}}  
+    
 4\pi L \|f\|_2
\right)
\frac1{n^{\frac {k}{k+3/2}}}
+ 
(2\|F'\|_1 
 + 
4\|F\|_\infty)
\frac{1}{n^{\frac {k+1/2}{k+3/2}}}
\end{align}
This concludes the proof.

\begin{remark}
    Notice that the term $2^{k+1}$ in the final estimate \eqref{eq:final_estimation} is sub-optimal in a lot of cases. For example, in case of smooth or analytic symbol and if $\|f^{(k)}\|_2$ converges to a constant $C$, then when we can minimise  \eqref{eq:final:bound_cnN} first with respect to $k$. Setting $N=\lfloor n^{\frac 1{k+3/2}}\rfloor$, we get  that $N^k\to n$, $N\to 1$, so we find the bound 
    \[
    \frac 1n\left( 
 \frac {2LC}{\sqrt{2\pi}} +    
2\|F'\|_1 
 + 
4\pi L\|f\|_2 + 4\|F\|_\infty
\right).
    \]
    
\end{remark}

\end{document}